\documentclass{article}
\pdfoutput=1
\PassOptionsToPackage{square,sort,comma,numbers}{natbib}
\usepackage[final]{neurips_2018}

\usepackage[utf8]{inputenc} % allow utf-8 input
\usepackage[T1]{fontenc}    % use 8-bit T1 fonts
\usepackage{hyperref}       % hyperlinks
\usepackage{url}            % simple URL typesetting
\usepackage{booktabs}       % professional-quality tables
\usepackage{amsfonts}       % blackboard math symbols
\usepackage{nicefrac}       % compact symbols for 1/2, etc.
\usepackage{microtype}      % microtypography

\title{Distributed $k$-Clustering for Data with Heavy Noise}

\usepackage{microtype}
\usepackage{subfig}

\usepackage{hyperref}

\usepackage{amsfonts}
\usepackage{amstext}
\usepackage{amsmath}
\usepackage{amssymb}
\usepackage{amsthm}

\usepackage{enumitem}
\usepackage{xspace}
\usepackage{algorithm}
\usepackage{algorithmic}

\usepackage{todonotes}

\usepackage{booktabs}
\usepackage{multirow}
\usepackage{natbib}

\usepackage{graphicx}
\usepackage{color}
\usepackage{float}
\newcommand{\etal}{\emph{et al.}\xspace}

\definecolor{lgray}{gray}{0.9}

\usepackage{bm}

\newcommand{\bbL}{\mathbb{L}}

\newcommand{\bbZ}{\mathbb{Z}}

\newcommand{\cost}{\mathsf{cost}}

\newcommand{\tldO}{\tilde{O}}
\DeclareMathOperator*{\argmin}{arg\,min}

\newtheorem{mythm}{Theorem}[section]
\newtheorem{mylmm}[mythm]{Lemma}
\newtheorem{mydef}[mythm]{Definition}
\newtheorem{myclm}[mythm]{Claim}

\usepackage{array}

\graphicspath{{figs/}}
\newcommand{\Z}{\mathbb{Z}}
\newcommand{\R}{\mathbb{R}}

\newcommand{\floor}[1]{{\left\lfloor#1\right\rfloor}}
\newcommand{\ceil}[1]{{\left\lceil#1\right\rceil}}
\newcommand{\set}[1]{{\left\{#1\right\}}}

\newcommand{\linebelow}{
	\vspace*{-5pt}
	
	\rule{\linewidth}{0.1pt}
}

\DeclareMathOperator*{\union}{\bigcup}

\mathchardef\mhyphen="2D

\newcommand{\poly}{\mathrm{poly}}
\newcommand{\ball}{{\mathsf{ball}}}
\newcommand{\seqkzc}{\ensuremath{\mathsf{kzc}}}
\newcommand{\distkzcmulti}{\ensuremath{\mathsf{dist\mhyphen kzc}}}

\newcommand{\seqgr}{\textsf{greedy}\xspace}
\newcommand{\distrr}{\textsf{random-random}\xspace}
\newcommand{\distrc}{\textsf{random-kzc}\xspace}
\newcommand{\MKCWM}{\textsf{MKCWM}\xspace}
\newcommand{\GLZ}{\textsf{GLZ}\xspace}
\newcommand{\GLZz}{\textsf{GLZ-z}\xspace}
\newcommand{\rgreedy}{\textsf{reverse-greedy}\xspace}
\newcommand{\kmeans}{\textsf{k-means}\xspace}
\newcommand{\kmeansmm}{\textsf{k-means$--$}\xspace}
\newcommand{\CAZ}{\textsf{CAZ}\xspace}
\newcommand{\BEL}{\textsf{BEL}\xspace}

\author{
  Xiangyu Guo \\
  University at Buffalo \\
  Buffalo, NY 14260 \\
  \texttt{xiangyug@buffalo.edu}    
  \And
  Shi Li \\
  University at Buffalo \\
  Buffalo, NY 14260 \\
  \texttt{shil@buffalo.edu}
}

\begin{document}

	\maketitle

	\begin{abstract}
	In this paper, we consider the $k$-center/median/means clustering with outliers problems (or the $(k, z)$-center/median/means problems) in the distributed setting.  Most previous distributed algorithms have their communication costs linearly depending on $z$, the number of outliers.  Recently Guha et al.\ \cite{DBLP:conf/spaa/GuhaLZ17} overcame this dependence issue by considering bi-criteria approximation algorithms that output solutions with $2z$ outliers.  For the case where $z$ is large, the extra $z$ outliers discarded by the algorithms might be too large, considering that the data gathering process might be costly. In this paper, we improve the number of outliers to the best possible $(1+\epsilon)z$, while maintaining the $O(1)$-approximation ratio and independence of communication cost on $z$.  The problems we consider include the $(k, z)$-center problem, and $(k, z)$-median/means problems in Euclidean metrics. Implementation of the our algorithm for $(k, z)$-center shows that it outperforms many previous algorithms, both in terms of the communication cost and quality of the output solution. 
\end{abstract}

\section{Introduction}
Clustering is a fundamental problem in unsupervised learning and data analytics.  In many real-life datasets, noises and errors unavoidably exist.  It is known that even a few noisy data points can significantly influence the quality of the clustering results. To address this issue, previous work has considered the clustering with outliers problem, where we are given a number $z$ on the number of outliers, and need to find the optimum clustering where we are allowed to discard $z$ points, under some popular clustering objective such as $k$-center, $k$-median and $k$-means.

Due to the increase in volumes of real-life datasets, and the emergence of modern parallel computation frameworks such as MapReduce and Hadoop, computing a clustering (with or without outliers) in the distributed setting has attracted a lot of attention in recent years. The set of points are partitioned into $m$ parts that are stored on $m$ different machines, who collectively need to compute a good clustering by sending messages to each other.  Often, the time to compute a good solution is dominated by the communications among machines. Many recent papers on distributed clustering have focused on designing $O(1)$-approximation algorithms with small communication cost \cite{DBLP:conf/nips/BalcanEL13, DBLP:conf/nips/MalkomesKCWM15, DBLP:conf/spaa/GuhaLZ17}. %TODO: more citations here.

Most previous algorithms for clustering with outliers have the communication costs linearly depending on $z$, the number of outliers.  Such an algorithm performs poorly when data is very noisy. Consider the scenario where distributed sensory data are collected by a crowd of people equipped with portable sensory devices.  Due to different skill levels of individuals and the quality of devices, it is reasonable to assume that a small constant fraction of the data points are unreliable.  

Recently, Guha et al.\ \cite{DBLP:conf/spaa/GuhaLZ17} overcame the linear dependence issue, by giving distributed $O(1)$-approximation algorithms for $k$-center/median/means with outliers problems with communication cost independent of $z$. However, the solutions produced by their algorithms have $2z$ outliers. Such a solution discards $z$ more points compared to the (unknown) optimum one, which may greatly decrease the efficiency of data usage. Consider an example where a research needs to be conducted using the inliers of a dataset containing 10\% noisy points; a filtering process is needed to  remove the outliers. A solution with $2z$ outliers will only preserve 80\% of data points, as opposed to the promised 90\%. As a result, the quality of the research result may be reduced.

Unfortunately, a simple example (described in the supplementary material) shows that if we need to produce any multiplicatively approximate solution with only $z$ outliers, then the linear dependence on $z$ can not be avoided. We show that, even deciding whether the optimum clustering with $z$ outliers has cost 0 or not, for a dataset distributed on 2 machines, requires a communication cost of $\Omega(z)$ bits. %This is irrespective of the objective function since in a cost-0 clustering under any common objective function, non-outliers are in at most $k$ different locations. Thus, if we are aiming for any multiplicative approximation, we need $\Omega(z)$ communication cost.
Given such a negative result and the positive results of Guha et al.\ \cite{DBLP:conf/spaa/GuhaLZ17}, the following question is interesting from both the practical and theoretical points of view:

\emph{Can we obtain distributed $O(1)$-approximation algorithms for $k$-center/median/means with outliers that have
communication cost independent of $z$ and output solutions with $(1 + \epsilon)z$ outliers, for any $\epsilon > 0$?}

On the practical side, an algorithm discarding $\epsilon z$ additional outliers is acceptable, as the number can be made arbitrarily small, compared to both the promised number $z$ of outliers and the number $n-z$ of inliers. On the theoretical side,  the $(1+\epsilon)$-factor for the number of outliers is the best we can hope for if we are aiming at an $O(1)$-approximation algorithm with communication complexity independent of $z$; thus answering the question in the affirmative can give the tight tradeoff between the number of outliers and the communication cost in terms of $z$.

In this paper, we make progress in answering the above question for many cases.  For the $k$-center objective, we solve the problem completely by giving a $(24(1+\epsilon), 1+\epsilon)$-bicriteria approximation algorithm with communication cost $O\left(\frac{km}{\epsilon}+\frac{m\log\Delta}{\epsilon}\right)$, where $\Delta$ is the aspect ratio of the metric. ($24(1+\epsilon)$ is the approximation ratio, $1+\epsilon$ is the multiplicative factor for the number of outliers our algorithm produces; the formal definition appears later.) For $k$-median/means objective, we give a distributed $(1+\epsilon, 1+\epsilon)$-bicrteria approximation algorithm for the case of Euclidean metrics. The communication complexity of the algorithm is $\textrm{poly}\left(\frac1\epsilon, k, D, m, \log \Delta\right)$, where $D$ is the dimension of the underlying Euclidean metric. (The exact communication complexity is given in Theorem~\ref{thm:kzmedian-main}.) Using dimension reduction techniques, we can assume $D = O(\frac{\log n}{\epsilon^2})$, by incurring a $(1+\epsilon)$-distortion in pairwise distances. So, the setting indeed covers a broad range of applications, given that the term ``$k$-means clustering'' is defined and studied exclusively in the context of Euclidean metrics. The $(1+\epsilon, 1+\epsilon)$-bicriteria approximation ratio comes with a caveat: our algorithm has running time exponential in many parameters such as $\frac 1\epsilon, k, D$ and $m$ (though it has no exponential dependence on $n$ or $z$).

%TODO: discussion improve the running time/ generalf metrics.

\subsection{Formulation of Problems}
We call the $k$-center (resp. $k$-median and $k$-means) problem with $z$ outliers as the $(k, z)$-center (resp. $(k, z)$-median and $(k, z)$-means)  problem. Formally, we are given a set $P$ of $n$ points that reside in a metric space $d$, two integers $k \geq 1$ and $z \in [0, n]$.  The goal of the problem is to find a set $C$ of $k$ centers and a set $P' \subseteq P$ of $n-z$ points so as to minimize $\max_{p \in P'}d(p, C)$ (resp. $\sum_{p \in P'} d(p, C)$ and $\sum_{p \in P'} d^2(p, C)$), where $d(p, C) = \min_{c \in C}d(p, c)$ is the minimum distance from $p$ to a center in $C$. For all the 3 objectives, given a set $C \subseteq  P$ of $k$ centers, the best set $P'$ can be derived from $P$ by removing the $z$ points $p \in P$ with the largest $d(p, C)$. Thus, we shall only use a set $C$ of $k$ centers to denote a solution to a $(k, z)$-center/median/means instance. The \emph{cost} of a solution $C$ is defined as $\max_{p \in P'} d(p, C)$, $\sum_{p \in P'} d(p, C)$ and $\sum_{p \in P'} d^2(p, C)$ respectively for a $(k, z)$-center, median and means instance, where $P'$ is obtained by applying the optimum strategy.  The $n-z$ points in $P'$ and the $z$ points in $P \setminus P'$ are called \emph{inliers} and \emph{outliers} respectively in the solution.

As is typical in the machine learning literature, we consider general metrics for $(k, z)$-center, and Euclidean metrics for $(k, z)$-median/means. In the $(k, z)$-center problem, we assume that each point $p$ in the metric space $d$ can be described using $O(1)$ words, and given the descriptions of two points $p$ and $q$, one can compute $d(p, q)$ in $O(1)$ time. In this case, the set $C$ of centers must be from $P$ since these are all the points we have.  For $(k, z)$-median/means problem, points in $P$ and centers $C$ are from Euclidean space $\R^D$, and it is not required that $C \subseteq P$. One should treat $D$ as a small number, since dimension reduction techniques can be applied to project points to a lower-dimension space. 

\noindent \textbf{Bi-Criteria Approximation}\quad We say an algorithm for the $(k, z)$-center/median/means problem achieves a bi-criteria approximation ratio (or simply approximation ratio) of $(\alpha, \beta)$, for some $\alpha, \beta \geq 1$, if it outputs a solution with at most $\beta z$ outliers, whose cost is at most $\alpha$ times the cost of the optimum solution with $z$ outliers.

\noindent \textbf{Distributed Clustering}\quad In the distributed setting, the dataset $P$ is split among $m$ machines, where $P_i$ is the set of data points stored on machine $i$. We use $n_i$ to denote $|P_i|$.  Following the communication model of \cite{DBLP:conf/icml/DingLHL16} and \cite{DBLP:conf/spaa/GuhaLZ17},  we assume there is a central coordinator, and communications can only happen between the coordinator and the $m$ machines.  The communication cost is measured in the total number of words sent.  Communications happen in rounds, where in each round, messages are sent between the coordinator and the $m$ machines. A message sent by a party (either the coordinator or some machine) in a round can only depends on the input data given to the party, and the messages received by the party in previous rounds. As is common in most of the previous results, we require the number of rounds used to be small, preferably a small constant. 

Our distributed algorithm needs to output a set $C$ of $k$ centers, as well as an upper bound $L$ on the maximum radius of the generated clusters.  For simplicity, only the coordinator needs to know $C$ and $L$. We do not require the coordinator to output the set of outliers since otherwise the communication cost is forced to be at least $z$. In a typical clustering task, each machine $i$ can figure out the set of outliers in its own dataset $P_i$ based on $C$ and $L$ (1 extra round may be needed for the coordinator to send $C$ and $L$ to all machines). 
	
%In a typical situation, we should think that $k, D$ and $m$ are small. $k$ is generally viewed as a small constant. Using the dimension reduction technique, one can assume that the dimension $d$ is small, say of order $O(\log n)$.  The number $m$ of machines can be slightly bigger, say, of order $n^{0.1}$.  To capture the high noise aspect of datasets, we should think of $z$ as a large number, e.g, of order $n^{0.9}$ or even $0.1n$.

%$\Delta = L_{\max}/L_{\min}$,  where $L_{\max}$ and $L_{\min}$ are upper and lower bounds on a non-zero optimum cost $L^*$ of the given instance. When we do not have extra information about the instance $P$, we simply define $L_{\max} = \max_{p, q \in P}d(p, q)$ and $L_{\min} = \frac{1}{2}\min_{p, q \in P:d(p,q) > 0}d(p, q)$. 

\subsection{Prior Work} %TODO: need to change and add literature for k-median/means.
In the centralized setting, we know the best possible approximation ratios of $2$ and $3$ \cite{DBLP:conf/soda/CharikarKMN01} for the $k$-center and $(k, z)$-center problems respectively, and thus our understanding in this setting is complete.  There has been a long stream of research on approximation algorithms $k$-median and $k$-means, leading to the current best approximation ratio of $2.675$ \cite{BPRST17} for $k$-median, $9$~\cite{ANSW16} for $k$-means, and $6.357$ for Euclidean $k$-means~\cite{ANSW16}.  The first $O(1)$-approximation algorithm for $(k, z)$-median is given by Chen, \cite{DBLP:conf/soda/Chen08}. Recently, Krishnaswamy et al.\ \cite{KLS18} developed a general framework that gives $O(1)$-approximations for both $(k, z)$-median and $(k, z)$-means.

%
%A simple greedy algorithm leads to a 2-approximation algorithm for $k$-center, and \cite{DBLP:conf/soda/CharikarKMN01} developed a 3-approximation algorithm for $(k, z)$-center. Both approximation ratios are the best possible assuming $\textrm{P} \neq \textrm{NP}$.
%
Much of the recent work has focused on solving $k$-center/median/means and $(k,z)$-center/median/means problems in the distributed setting \cite{DBLP:conf/kdd/EneIM11, DBLP:conf/nips/BalcanEL13, DBLP:conf/spaa/ImM15, DBLP:conf/nips/MalkomesKCWM15, DBLP:conf/spaa/ImM15, DBLP:conf/nips/MalkomesKCWM15, DBLP:conf/icml/DingLHL16, DBLP:conf/nips/ChenSWZ16, DBLP:conf/spaa/GuhaLZ17, DBLP:journals/corr/abs-1805-09495}.  Many distributed $O(1)$ approximation algorithms with small communication complexity are known for these problems. However, for $(k, z)$-center/median/means problems, most known algorithms have communication complexity linearly depending on $z$, the number of outliers. Guha et al.\ \cite{DBLP:conf/spaa/GuhaLZ17} overcame the dependence issue, by giving $(O(1), 2+\epsilon)$-bicriteria approximation algorithms for all the three objectives. The communication costs of their algorithms are $\tilde O(m/\epsilon + mk)$, where $\tilde O$ hides a logarithmic factor.

\subsection{Our Contributions} %In all the previous distributed algorithms for the $(k, z)$-center problem (e.g, those in \cite{DBLP:conf/spaa/ImM15}, \cite{DBLP:conf/nips/MalkomesKCWM15} and \cite{DBLP:conf/spaa/GuhaLZ17}), there is a linear dependence of the communication cost on the number $z$ of outliers. As we argued, this dependence can not be removed if one is seeking for a true approximation algorithm. Our main contribution is in developing a distributed bi-criteria approximation algorithm with communication cost independent of $z$ and $n$, that give $O(1)$-approximation with slightly more than $z$ outliers. 
Our main contributions are in designing $(O(1), 1+\epsilon)$-bicriteria approximation algorithms for the $(k, z)$-center/median/means problems. The algorithm for $(k, z)$-center works for general metrics:
%For the $(k, z)$-center problem, we give a distributed algorithm with communication cost $O(\frac{km}{\epsilon}\cdot \frac{\log \Delta}{\epsilon})$, that achieves an $(O(1), 1+\epsilon)$-approximation. That is, we achieve a communication cost independent of $n$ and $z$, by allowing our solutions to have slightly more than $z$ outliers.
\begin{mythm}
	\label{thm:main}
	There is a $4$-round, distributed algorithm for the $(k, z)$-center problem, that achieves a $(24(1+\epsilon), 1+\epsilon)$-bicriteria approximation and $O\left(\frac{km}{\epsilon}+\frac{m\log\Delta}\epsilon\right)$ communication cost, where $\Delta$ is the aspect ratio of the metric. %The algorithm needs to know the cost $L^*$ of the optimum solution (with $z$ outliers).
\end{mythm}

%The communication cost of our algorithm is only $O\left(\frac{km}{\epsilon}\right)$, polynomial in the small parameters $k, m, d$ and $1/\epsilon$. As a result, it can be easily generalized to the case where points have weights and we can allow a total weight $z$ of outliers. In such a case, $z$ might be much larger than $n$ and previous algorithms require large communication costs.

We give a high-level picture of the algorithm. By guessing, we assume that we know the optimum cost $L^*$ (since we do not know, we need to lose the additive $\frac{m\log\Delta}\epsilon$ term in the communication complexity). In the first round of the algorithm, each machine $i$ will call a procedure called \textsf{aggregating}, on its set $P_i$. This procedure performs two operations. First, it discards some points from $P_i$; second, it moves each of the survived points by a distance of at most $O(1)L^*$.  After the two operations, the points will be \emph{aggregated} at a few locations.  Thus, machine $i$ can send a compact representation of these points to the coordinator: a list of $(p, w'_p)$ pairs, where $p$ is a location and $w'_p$ is the number of points aggregated at $p$.   The coordinator will collect all the data points from all the machines, and run the algorithm of \cite{DBLP:conf/soda/CharikarKMN01} for $(k, z')$-center instance on the collected points, for some suitable $z'$.

To analyze the algorithm, we show that the set $P'$ of points collected by the coordinator well-approximates the original set $P$.  The main lemma is that the total number of non-outliers removed by the aggregation procedure on all machines is at most $\epsilon z$.  This incurs the additive factor of $\epsilon z$ in the number of outliers.  We prove this by showing that inside any ball of radius $L^*$, and for every machine $i \in [m]$, we removed at most $\frac{\epsilon z}{km}$ points in $P_i$. Since the non-outliers are contained in the union of $k$ balls of radius $L^*$, and there are $m$ machines, the total number of removed non-outliers is at most $\epsilon z$.  For each remaining point, we shift it by a distance of $O(1)L^*$, leading to an $O(1)$-loss in the approximation ratio of our algorithm. 

We perform experiments comparing our main algorithm stated in Theorem~\ref{thm:main} with many previous ones on real-world datasets. The results show that it matches the state-of-art method in both solution quality (objective value) and communication cost. We remark that the qualities of solutions are measured w.r.t removing only $z$ outliers. Theoretically, we need $(1+\epsilon)z$ outliers in order to achieve an $O(1)$-approximation ratio and our constant 24 is big. In spite of this, empirical evaluations suggest that the algorithm on real-word datasets performs much better than what can be proved theoretically in the worst case. 

For $(k, z)$-median/means problems, our algorithm works for the Euclidean metric case and has  communication cost depending on the dimension $D$ of the Euclidean space. One can w.l.o.g. assume $D = O(\log n/\epsilon^2)$ by using the dimension reduction technique. Our algorithm is given in the following theorem:
\begin{mythm} \label{thm:kzmedian-main}
	There is a $2$-round, distributed algorithm for the $(k, z)$-median/means problems in $D$-dimensional Euclidean space, that achieves a $(1+\epsilon, 1+\epsilon)$-bicriteria approximation ratio with probability $1-\delta$. The algorithm has communication cost $O\left(\Phi D \cdot \frac{\log (n\Delta/\epsilon)}{\epsilon}\right)$, where $\Delta$ is the aspect ratio of the input points, $\Phi = O\left(\frac{1}{\epsilon^2}(kD + \log\frac1\delta) + mk\right)$ for $(k, z)$-median, and $\Phi = O\left(\frac{1}{\epsilon^4}(kD + \log\frac1\delta)+mk\log\frac{mk}{\delta}\right)$ for $(k, z)$-means.
\end{mythm}

We now give an overview of our algorithm for $(k,z)$-median/means. First, it is not hard to reformulate the objective of the $(k, z)$-median problem as minimizing $\sup_{L\geq 0} \left(\sum_{p \in P} d_L(p, C) - zL\right)$, where $d_L$ is obtained from $d$ by truncating all distances at $L$.  By discretization, we can construct a set $\bbL$ of $O\left(\frac{\log (\Delta n/\epsilon)}{\epsilon}\right)$ interesting values that the $L$ under the superior operator can take. Thus, our goal becomes to find a set $C$, that is simultaneously good for every $k$-median instance defined by $d_L, L \in \bbL$.  Since now we are handling $k$-median instances (without outliers), we can use the communication-efficient algorithm of \cite{DBLP:conf/nips/BalcanEL13} to construct an $\epsilon$-coreset $Q_L$ with weights $w_L$ for every $L \in \bbL$. Roughly speaking, the coreset $Q_L$ is similar to the set $P$ for the task of solving the $k$-median problem under metric $d_L$. The size of each $\epsilon$-coreset $Q_L$ is at most $\Phi$, implying the communication cost stated in the theorem. After collecting all the coresets, the coordinator can approximately solve the optimization problem on them. This will lead to an $(1+O(\epsilon), 1+O(\epsilon))$-bicriteria approximate solution.  The running time of the algorithm, however, is exponential in the total size of the coresets.  The argument can be easily adapted to the $(k, z)$-means setting.

\textbf{Organization}\quad In Section~\ref{sec:kz-center-multiplicative}, we prove Theorem~\ref{thm:main}, by giving the $(24(1+\epsilon), 1+\epsilon)$-approximation algorithm. The empirical evaluations of our algorithm for $(k, z)$-center and the proof of Theorem~\ref{thm:kzmedian-main} are provided in the supplementary material.

\paragraph{Notations}
Throughout the paper, point sets are multi-sets, where each element has its own identity.  By a copy of some point $p$, we mean a point with the same description as $p$ but a different identity. For a set $Q$ of points, a point $p$,  and a radius $r \geq 0$, we define $\ball_Q(p, r) = \set{q \in Q: {d}(p, q) \leq r}$ to be the set of points in $Q$ that have distances at most $r$ to $p$. For a weight vector $w \in \Z_{\geq 0}^Q$ on some set $Q$ of points, and a set $S \subseteq Q$, we use $w(S) = \sum_{p \in S}w_p$ to denote the total weight of points in $S$. 

Throughout the paper, $P$ is always the set of input points. We shall use $d_{\min} = \min_{p, q \in P: d(p, q) > 0} d(p, q)$ and $d_{\max} = \max_{p, q \in P} d(p, q)$ to denote the minimum and maximum non-zero pairwise distance between points in $P$. Let $\Delta = \frac{d_{\max}}{d_{\min}}$ denote the aspect ratio of the metric.

	\section{Distributed $(k, z)$-Center Algorithm with $(1+\epsilon)z$ Outliers}
\label{sec:kz-center-multiplicative}

In this section, we prove Theorem~\ref{thm:main}, by giving the $(24(1+\epsilon), 1+\epsilon)$-approximation algorithm for $(k, z)$-center, with communication cost $O\left(\frac{km}{\epsilon} + \frac{m\log \Delta}{\epsilon}\right)$.  Let $L^*$ be the cost of the optimum $(k, z)$-center solution (which is not given to us). We assume we are given a parameter $L \geq 0$ and our goal is to design a main algorithm with communication cost $O\left(\frac{km}{\epsilon}\right)$, that either returns a $(k, (1+\epsilon)z)$-center solution of cost at most $24L$, or certifies that $L^* > L$.  Notice that $L^* \in \set{0} \cup [d_{\min}/2, d_{\max}]$. We can obtain our $(24(1+\epsilon), 1+\epsilon)$-approximation by using the main algorithm to check $O\left(\frac{\log \Delta}{\epsilon}\right)$ different values of $L$ in parallel, and among all $L$'s that are \emph{not} certified to be less than $L^*$, returning solution correspondent to the smallest such $L$.  A naive implementation requires all the parties to know $d_{\min}$ and $d_{\max}$ in advance; we show in the supplementary material that the requirement can be removed.

In intermediate steps, we may deal with $(k, z)$-center instances where points have integer weights. In this case, the instance is defined as $(Q, w)$, where $Q$ is a set of points, $w \in \Z_{>0}^Q$, and $z$ is an integer between $0$ and $w(Q)=\sum_{q \in Q}w_q$. The instance is equivalent to the instance $\hat Q$, the multi-set where we have $w_q$ copies of each $q \in Q$.

\cite{DBLP:conf/soda/CharikarKMN01} gave a $3$-approximation algorithm for the $(k, z)$-center problem. However, our setting is slightly more general so we can not apply the result directly. We are given a weighted set $Q$ of points that defines the $(k, z)$-center instance. The optimum set $C^*$ of centers, however, can be from the superset $P \supseteq Q$ which is hidden to us. Thus, our algorithm needs output a set $C$ of $k$ centers from $Q$ and compare it against the optimum set $C^*$ of centers from $P$. Notice that by losing a factor of $2$, we can assume centers are in $Q$; this will lead to a $6$-approximation.  Indeed, by applying the framework of \cite{DBLP:conf/soda/CharikarKMN01} more carefully, we can obtain a $4$-approximation for this general setting. We state the result in the following theorem:
\begin{mythm}[\cite{DBLP:conf/soda/CharikarKMN01}]\label{thm:charikar}
	Let $d$ be a metric over the set $P$ of points, $Q \subseteq P$ and $w \in \Z_{>0}^Q$. There is an algorithm $\seqkzc$ (Algorithm~\ref{alg:charikar}) 
 that takes inputs $k, z' \geq 1$, $(Q, w')$ with $|Q| = n'$, the metric $d$ restricted to $Q$, and a real number $L' \geq 0$. In time $O(n'^2)$, the algorithm either outputs a $(k, z')$-center solution $C' \subseteq Q$ to the instance $(Q, w')$ of cost at most $4L'$, or certifies that there is no $(k, z')$-center solution $C^* \subseteq P$ of cost at most $L'$ and outputs ``No''.
\end{mythm}

The main algorithm is $\distkzcmulti$ (Algorithm~\ref{alg:kzcenter}), which calls an important procedure called $\mathsf{aggregating}$ (Algorithm~\ref{alg:aggregating}).  We describe $\mathsf{aggregating}$ and $\distkzcmulti$ in Section~\ref{subsec:aggregation} and \ref{subsec:24-approx} respectively.

%TODO: improving running time.

\subsection{Aggregating Points}
\label{subsec:aggregation}

The procedure $\mathsf{aggregating}$, as described in Algorithm~\ref{alg:aggregating}, takes as input the set $Q \subseteq P$  of points to be aggregated (which will be some $P_i$ when we actually call the procedure), the guessed optimum cost $L$, and $y \geq 0$, which controls how many points can be removed from $Q$. It returns a set $Q'$ of points obtained from aggregation, along with their weights $w'$. 

\noindent \begin{minipage}[t]{0.53\textwidth}
	\vspace{-20pt}
	\begin{algorithm}[H]
	  \caption{$\seqkzc(k, z', (Q, w'), L')$} \label{alg:charikar}
	  \begin{algorithmic}[1]
	    \STATE $U\leftarrow Q, C' \gets \emptyset$;
	    \STATE \textbf{for} {$i \gets 1$ to $k$} \textbf{do}
	    \STATE \hspace*{\algorithmicindent} $p_i \gets p \in Q$ with largest $w'(\ball_U(p, 2L'))$
	    \STATE \hspace*{\algorithmicindent} $C' \gets C' \cup \set{p_i}$\;
	    \STATE \hspace*{\algorithmicindent} $U\leftarrow U \setminus \ball_{U}(p_i, {4}L')$ \; 
	    \STATE \textbf{if} {$w'(U) > z'$} \textbf{then return} ``No'' \textbf{else} {\bfseries return} $C'$\;
	  \end{algorithmic}
	\end{algorithm}
\end{minipage}\ %
\begin{minipage}[t]{0.47\textwidth}
	\vspace{-20pt}
	\begin{algorithm}[H]
		\caption{$\mathsf{aggregating}(Q, L, y)$} \label{alg:aggregating}
		\begin{algorithmic}[1]
	      \STATE $U\leftarrow Q, Q' \gets \emptyset$;
	      \STATE \textbf{while} {$\exists p \in Q$ with $|\ball_U(p, 2L)| > y$} \textbf{do}
		      \STATE \hspace*{\algorithmicindent} $Q' \gets Q' \cup \set{p}$,
		      %\label{step:aggregating-update}\;
		     % \STATE \hspace*{\algorithmicindent} 
		     $w'_p\leftarrow \left|\ball_{U}(p, {4}L)\right|$ \label{step:aggregating-update}
		      %\STATE \hspace*{\algorithmicindent} \label{step:aggregating-update}   
		      \STATE \hspace*{\algorithmicindent} $U\leftarrow U \setminus \ball_{U}(p, {4}L)$ \; 
	      \STATE \textbf{return} $(Q', w')$
	    \end{algorithmic}
	\end{algorithm}
\end{minipage}

In $\mathsf{aggregating}$, we start from $U = Q$ and $Q' = \emptyset$ and keep removing points from $U$.  In each iteration, we check if there is a $p \in Q$ with $|\ball_U(p, 2L)| \geq y$. If yes, we add $p$ to $Q'$, remove $\ball_U(p, 4L)$ from $U$ and let $w_p$  be the number of points removed.  We repeat thie procedure until such a $p$ can not be found.  We remark that the procedure is very similar to the algorithm $\seqkzc$ (Algorithm~\ref{alg:charikar}) in \cite{DBLP:conf/soda/CharikarKMN01}.

%One might notice that the procedure is very similar to that of \cite{DBLP:conf/soda/CharikarKMN01} for approximating the $(k, z)$-center problem. The one difference is that our algorithm only needs to find a point $p$ with $|\ball_U(p, 2L)|\geq y$ in every iteration, while the algorithm of \cite{DBLP:conf/soda/CharikarKMN01} needs to find the $p$ with the largest $|\ball_U(p, L)|$.

We start from some simple observations about the algorithm.
\begin{myclm}
	\label{claim:aggregating}
	We define $V = \union_{p \in  Q'}\ball_Q(p, {4}L)$ to be the set of points in $Q$ with distance at most ${4}L$ to some point in $Q'$ at the end of Algorithm~\ref{alg:aggregating}. Then,  the following statements hold at the end of the algorithm:
	\begin{enumerate}[itemsep=0pt,topsep=0pt,leftmargin=*]
%		\item $Q' \subseteq Q$.
		\item $U = Q \setminus V$.
		\item $\big|\ball_{U}(p, {2} L)\big| \leq y$ for every $p \in Q$.
		\item There is a function $f: V \to Q'$ such that $d(p, f(p)) \leq {4}L$, $\forall p \in V$, and $w'(q) = |f^{-1}(q)|, \forall q \in Q'$.
	\end{enumerate}
\end{myclm}
\begin{proof}
	$U$ is exactly the set of points in $Q$ with distance more than $4L$ to any point in $Q'$ and thus $U  = Q \setminus V$.  Property 2 follows from the termination condition of the algorithm. Property 3 holds by the way we add points to $Q'$ and remove points from $U$. If in some iteration we added $q$ to $Q'$, we can define $f(p) = q$ for every point $p \in \ball_U(p, 4L)$, i.e, every point removed from $U$ in the iteration.
\end{proof}

We think of $U$ as the set of points we discard from $Q$ and $V$ as the set of survived points. We then move each $p \in V$ to $f(p) \in Q'$ and thus $V$ will be aggregated at the set $Q'$ of locations.   The following crucial lemma upper bounds $|Q'|$:
\begin{mylmm}
	\label{lemma:Q'-small}
	Let $\hat z \geq 0$ and assume there is a $(k, \hat z)$-center solution 
	$C^* \subseteq P$ to the instance $Q$ with cost at most ${} L$.
	Then, at the end of Algorithm~\ref{alg:aggregating} we have $|Q'| \leq k + \frac{\hat z}{y}$.
\end{mylmm}
\begin{proof}
	Let $O= Q\setminus \union_{c \in C^*}\ball_Q(c, {} L)$ be the set of outliers according to solution $C^*$. % By the definition of values of $(k, z)$-center solutions, we have %$w\left(Q \setminus \union_{c \in C}\ball_Q(c, {2} L)\right) \leq \hat z$ and 
	%$w\left(Q\setminus \union_{c \in C}\ball_Q(c, {} L)\right) \leq \hat z$. 
	Thus $|O| \leq \hat z$.

	Focus on the moment before we run Step~\ref{step:aggregating-update} in some iteration of $\mathsf{aggregating}$. See Figure~\ref{fig:aggregating} for the two cases we are going to consider. In case (a), every center $c \in \ball_{C^*}(p, 3L)$ has $\ball_{U}(c, {} L) = \emptyset$. In this case, every point $q \in \ball_U(p,{2} L)$ has $d(q, C^*) > L$: if $d(p, c) > {3}L$ for some $c \in C^*$, then $d(q, c) \geq d(p, c) - d(p, q) > 3L - 2L = L$ by triangle inequality; for some $c \in C^*$ with $d(p, c) \leq {3}L$, we have $\ball_U(c, L) = \emptyset$, implying that $d(q, c) > L$ as $q \in U$. Thus, $\ball_{U}(p, {2} L) \subseteq O$. So, Step~\ref{step:aggregating-update} in this iteration will decrease $|O \cap U|$ by at least $|\ball_U(p, 4L)| \geq |\ball_U(p, 2L)| > y$. 
	
	Consider the case (b) where some $c \in \ball_{C^*}(p, 3L)$ has $\ball_{U}(c, {} L) \neq \emptyset$. Then $\ball_U(p, 4L) \supseteq \ball_{U}(c, {} L)$ will be removed from $U$ by Step~\ref{step:aggregating-update} in this iteration. Thus, 
	\begin{enumerate}[topsep=0pt,itemsep=0pt,leftmargin=*]
		\item if case (a) happens, then $|U \cap O|$ is decreased by more than $y$ in this iteration;
		\item otherwise case (b) happens; then for some $c \in C^*$, $\ball_U(c, {} L)$ changes from non-empty to $\emptyset$.
	\end{enumerate}
	The first event can happen for at most $|O|/y \leq \hat z/y$ iterations and the second event can happen for at most $|C^*| \leq k$ iterations. So, $|Q'| \leq k + \hat z/y$. 
\end{proof} \vspace*{-10pt}

\begin{figure}
	\centering
	\noindent \begin{minipage}{0.52\textwidth}
		\includegraphics[width=0.96\textwidth]{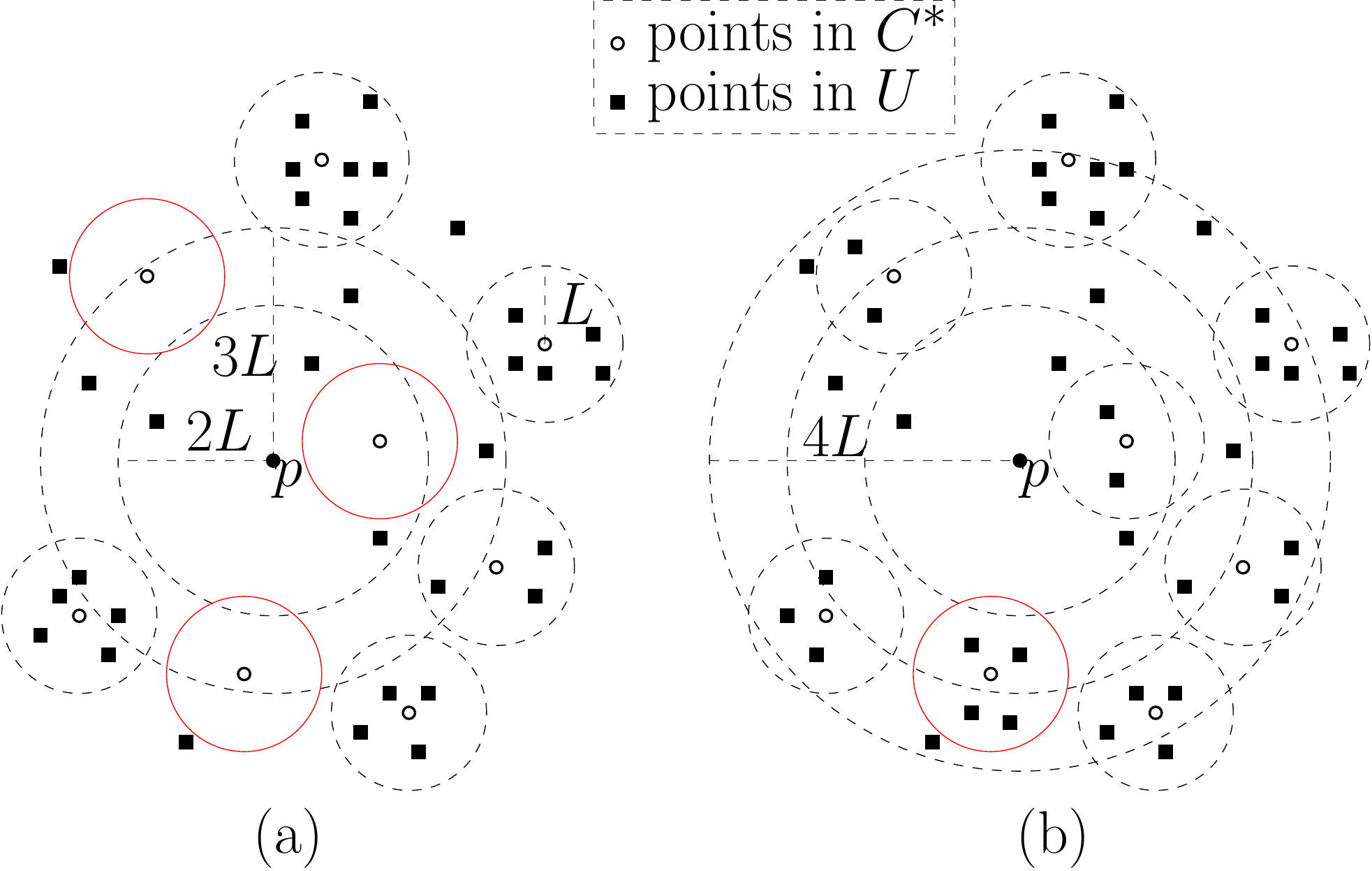}
		\caption{Two cases in proof of Lemma~\ref{lemma:Q'-small}. In Figure (a), the balls $\set{\ball_U(c, L): c \in C^*, d(p, c) \leq 3L}$ (red circles) are all empty. So, $\ball_{U}(p, 2L) \subseteq O$. In Figure (b), there is a non-empty $\ball_U(c, L)$ for some $c \in C^*$ with $d(p, c) \leq 3L$ (the red circle).  The ball is contained in $\ball_U(p, 4L)$.
		}
		\label{fig:aggregating}
	\end{minipage} \quad
	\begin{minipage}{0.44\textwidth}
		 \centering
		 \includegraphics[width=0.96\textwidth]{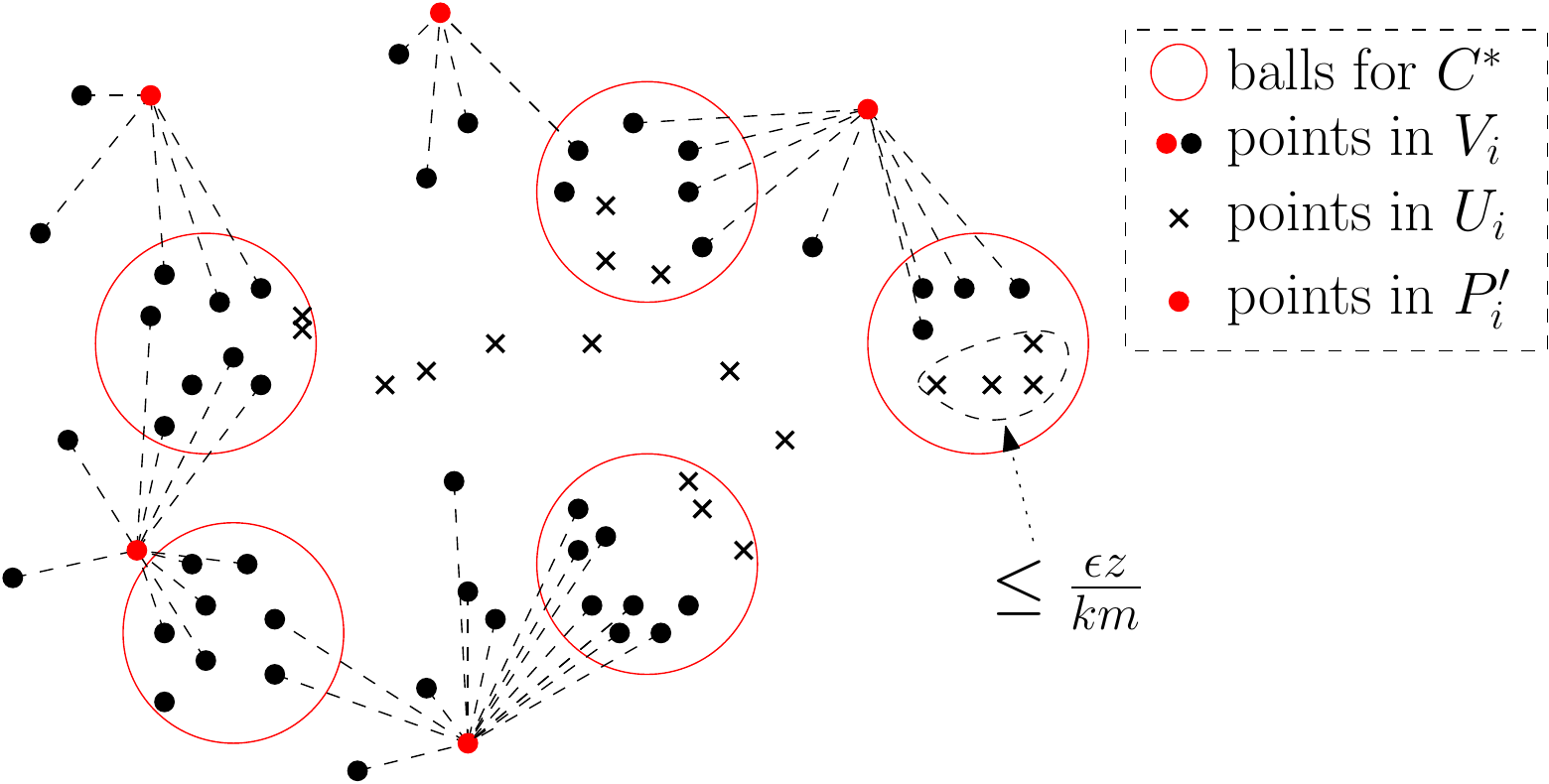}
		 \caption{Illustration for proof of Lemma~\ref{lemma:kzcenter-approx}. %The figure is for one machine $i$. %The optimum solution $C^*$ is indicatd by red circles. Dots and crosses represent points in  some $P_i$, where dots represent points in $V_i$ and crosses represent points in $U_i$. $P'_i$ are indicated by red dots and 
		 %The mapping 
		 $f_i:V_i \to P'_i$ is indicated by the dashed lines, each of whom is of length at most $4L$. The number of crosses in a circle is at most $\frac{\epsilon z}{km}$.%; thus the total number of crosses in all the $k$ circles, accross all the $m$ machines, is at most $\epsilon z$.
		 } 
		 \label{fig:aggregating1}
	\end{minipage}
\end{figure} 

\subsection{The Main Algorithm}
\label{subsec:24-approx}
	We are now ready to describe the main algorithm for the $(k, z)$-center problem, given in Algorithm~\ref{alg:kzcenter}. In the first round, each machine will call $\mathsf{aggregating}(P_i, L, \frac{\epsilon z}{km})$ to obtain $(P'_i, w'_i)$.  All the machines will first send their corresponding $|P'_i|$ to the coordinator. In Round 2 the algorithm will check if $\sum_{i \in [m]}|P'_i|$ is small or not. If yes, send a ``Yes'' message to all machines; otherwise return ``No'' and terminate the algorithm. In Round 3, if a machine $i$ received a ``Yes'' message from the coordinator, then it sends the dataset $P'_i$ with the weight vector $w'_i$ to the coordinator.  Finally in Round 4, the coordinator collects all the weighted points $P' = \union_{i \in [m]}P'_i$ and run $\seqkzc$ on these points. 

\begin{algorithm}[ht]
  \caption{$\distkzcmulti$}\label{alg:kzcenter}
  \begin{small}
      \textbf{input on all parties}: $n, k$, $z$, $m$, $L$, $\epsilon$\\
      \textbf{input on machine $i$}: dataset $P_i$ with $|P_i| = n_i$\\
      \textbf{output}: a set $C' \subseteq P$ or ``No'' (which certifies $L^* > L$)  \linebelow
      
     \textbf{Round 1 on  machine $i \in [m]$}
     
     %TODO: alpha and beta not needed.
     
  	\begin{algorithmic}[1]
%	  	\STATE Let $F_i$ be either $P_i$, or the subset of $P_i$ constructed in Section~\ref{subsec:improve-running-time} (described later)\; 
  		\STATE $(P'_i, w'_i) \gets \mathsf{aggregating}(P_i, L, \frac{\epsilon z}{km})$\; \label{step:kzcenter-call-aggregating}
		\STATE send $|P'_i|$ to the coordinator\;
	\end{algorithmic}\linebelow
	
	\textbf{Round 2 on the coordinator}
	
	\begin{algorithmic}[1]
		\STATE \textbf{if} $\sum_{i \in [m]}|P'_{i}| > km(1+1/\epsilon)$ \textbf{then} \textbf{return} ``No'' \textbf{else} send ``Yes'' to each machine $i \in [m]$
	\end{algorithmic}\linebelow
	
	\textbf{Round 3 on  machine $i \in [m]$}
	
	\begin{algorithmic}[1]			
		\STATE Upon receiving of a ``Yes'' message from the coordinator, respond by sending $(P'_i, w'_i)$
	\end{algorithmic}\linebelow
	
	\textbf{Round 4 on the coordinator}
	
	\begin{algorithmic}[1]
		\STATE let $P' \gets \union_{i = 1}^mP'_i$ 
		\STATE let $w'$ be the function from $P'$ to $\bbZ_{>0}$ obtained by merging $w'_1, w'_2, \cdots, w'_m$\;
		\STATE let $z' \gets (1+\epsilon) z + w'(P') - n$\;
		\STATE \textbf{if} {$z'  < 0$} \textbf{then return} ``No'' \textbf{else} {\bfseries return} $\seqkzc(k, z', (P', w'), L'=5L)$\; \label{Step:distributed-returning}
	\end{algorithmic}
  \end{small}
\end{algorithm}

An immediate observation about the algorithm is that its communication cost is small:	
\begin{myclm}
	The communication cost of $\distkzcmulti$ is $O(\frac{km}{\epsilon})$.
\end{myclm}
\vspace*{-15pt}

\begin{proof}
	The total communication cost of Round 1 and Round 2 is $O(m)$.  We run Round 3 only when the coordinator sent the ``Yes'' message, in which case the communication cost is at most $\sum_{i = 1}^m |P'_i| \leq km(1+1/\epsilon) = O(\frac{km}{\epsilon})$.
\end{proof}\vspace*{-5pt}

	It is convenient to define some notations before we make further analysis. For every machine $i \in [m]$, let $P'_i$ be the $P'_i$ constructed in Round 1 on machine $i$. Let $V_i = \union_{p \in P'_i} \ball_{P_i}(p, 4L)$ be the set of points in $P_i$ that are within distance at most $4L$ to some point in $P'_i$. Notice that this is the definition of $V$ in Claim~\ref{claim:aggregating} for the execution of $\mathsf{aggregating}$ on machine $i$. Let $U_i = P_i \setminus V_i$; this is the set $U$ at the end of this execution. Let $f_i$ be the mapping from $V_i$ to $P'_i$ satisfying Property 3 of Claim~\ref{claim:aggregating}. Let $V = \union_{i \in [m]} V_i, P' = \union_{i \in [m]}P'_i$ and $f$ be the function from $V$ to $P'$, obtained by merging $f_1, f_2, \cdots, f_m$. Thus $(p, f(p)) \leq 4L, \forall p \in V$ and $w'(q) = |f^{-1}(q)|, \forall q \in P'$.  
	
\begin{myclm}
	If $\distkzcmulti$ returns a set $C'$, then $C'$ is a $(k, (1+\epsilon)z)$-center solution to the instance $P$ with cost at most $24L$. 
\end{myclm}
\begin{proof}
	$C'$ must be returned in Step~\ref{Step:distributed-returning} in Round 4. By Theorem~\ref{thm:charikar} for $\seqkzc$, $C'$ is a $(k, z')$-center solution to $(P', w')$ of cost at most $4 \cdot 5L = 20L$.  That is, $w'\left(P' \setminus \union_{c \in C'}\ball_{P'}(c, 20L)\right) \leq z'$. This implies $w'\left( \union_{c \in C'}\ball_{P'}(c, 20L)\right) \geq w'(P') - z' = n - (1+\epsilon)z$.  Notice that for each  $q \in P'$, the set $f^{-1}(q) \subseteq V \subseteq P$ of points are within distance $4L$ from $q$ and $w'(q) = |f^{-1}(q)|$. So, $\left|\union_{c \in C'}\ball_{P}(c, 24L)\right| \geq n - (1+\epsilon)z$, which is exactly $\left| P \setminus \union_{c \in C'} \ball_P(c, 24L)\right| \leq (1+\epsilon)z$.
\end{proof}

We can now assume $L \geq L^*$  and we need to prove that we must reach Step~\ref{Step:distributed-returning} in Round 4 and return a set $C'$. We define $C^* \subseteq P$ to be a set of size $k$ such that $|P \setminus \union_{c \in C^*}\ball(c, L)| \leq z$. Let $I = \union_{c \in C^*}\ball_P(c, L)$ be the set of  ``inliers'' according to $C^*$ and $O = P \setminus I$ be the set of outliers. Thus, $|I| \geq n - z$ and $|O| \leq z$.

\begin{mylmm}
	After Round 1, we have $\sum_{i \in [m]}|P'_i| \leq km(1+1/\epsilon)$.
\end{mylmm}
\vspace*{-15pt}

\begin{proof}
	Let $z_i = |P_i \cap O| = \Big|P_i \setminus \union_{c \in C^*}\ball_{P_i}(c, L)\Big|$ be the set of outliers in $P_i$. Then, $C^*$ is a $(k, z_i)$-center solution to the instance $P_i$ with cost at most $L$. 
	By Lemma~\ref{lemma:Q'-small}, we have that $|P'_i| \leq k + \frac{z_i}{\epsilon z/(km)}$. So, we have
	\begin{flalign*}
		&&\textstyle \sum_{i \in [m]}|P'_i| \leq km + \frac{km}{\epsilon z}\sum_{i \in [m]}z_i  \leq km\left(1+\frac1\epsilon\right). && \qedhere
	\end{flalign*}
\end{proof} \vspace*{-5pt}

Therefore, the coordinator will not return ``No'' in Round 2. It remains to prove the following Lemma.
\begin{mylmm}\label{lemma:kzcenter-approx}
	%Assume the set $F_i$ chosen by machine $i$ in Step 1 of round 1 has the following property: For every $c \in C^*$ with $\ball_{P_i}(c, L) \neq \emptyset$, we have $F_i \cap \ball_{P_i}(c, L) \neq \emptyset$. 
	Algorithm~\ref{alg:kzcenter} will reach Step~\ref{Step:distributed-returning} in Round 4 and return a set $C'$. %, and return a set $C' \subseteq \R^$ of size at most $k$ such that  $w'\left(P'\setminus \union_{c \in C'}\ball_{P'}(c, 20L)\right) \leq z'$.
	
%	$\left|P\setminus \union_{c \in C'}\ball_P(c, 24L)\right| \leq (1+\epsilon)z$. %and give $k$ centers that cover at least $n-(1+\epsilon)z$ points with cost $\,9L$; in addition, the communication cost is $O(mk/\epsilon)$.
\end{mylmm}

\vspace*{-10pt}
 
\begin{proof}
	See Figure~\ref{fig:aggregating1} for the illustration of the proof. By Property 2 of Claim~\ref{claim:aggregating}, we have $\left|\ball_{U_i}(p, 2 L)\right| \leq \frac{\epsilon z}{km}$ for every $p \in U_i$ since $U_i \subseteq P_i$.  This implies that for every $c \in C^*$, we have $\left|\ball_{U_i}(c, L)\right| \leq \frac{\epsilon z}{km}$. (Otherwise, taking an arbitrary $p$ in the ball leads to a contradiction.)
	\begin{align*}
		& \quad |U_i \cap I| = \Big|\union_{c \in C^*}\ball_{U_i}(c, L)\Big| \leq \sum_{c \in C^*}|\ball_{U_i}(c, L)| \leq \sum_{c \in C^*}\frac{\epsilon z}{km} \leq \frac{\epsilon z}{m}, \quad \forall i \in [m].\\
%	\end{align*}
%	Taking all machines $i\in[m]$ into consideration, we have
%	\begin{align*}
		& \quad \sum_{i \in [m] } |I \cap V_i|
		 = \sum_{i \in [m] }\big(|I \cap P_i| - |I \cap U_i|\big)  \geq \sum_{i \in [m] } \left(|I \cap P_i| - \frac{\epsilon z}{m}\right)  
		                                                      = |I| - \epsilon  z \geq n - (1+\epsilon)z.
	\end{align*}
 
  For every $p \in V \cap I$, $f(p)$ will have distance at most $L + 4L = 5L$ to some center in $C^*$. Also, notice that $w'(q) = |f^{-1}(q)|$ for every $q \in P'$, we have that 
  \begin{align*}
  	\textstyle w'\big(\union_{c \in C^*}\ball_{P'}(c, 5L)\big) \geq %\\
 % 	&= \left|f^{-1}\left(\union_{c \in C^*}\ball_{P'}(c, 5L)\right)\right|\\
  	%&\geq \left|\union_{c \in C^*}\ball_{V}(c, L)\right| = 
  	|V \cap I| \geq n - (1+\epsilon) z.
  \end{align*}
  So, $w'(P' \setminus \union_{c \in C^*}\ball_{P'}(c, 5L)) \leq w(P') - n + (1+\epsilon)z = z'$. This implies that $z' \geq 0$, and there is a $(k, z')$-center solution $C^* \subseteq P$ to the instance $(P', w')$ of cost at most $5L$. Thus $\distkzcmulti$ will reach Step~\ref{Step:distributed-returning} in Round 4 and returns a set $C'$. %such that $|C'| \leq k$ and $w'\left(P' \setminus \union_{c \in C'}\ball_{P'}(c, 20L)\right) \leq z'$.   
  This finishes the proof of the Lemma.
\end{proof}\vspace*{-5pt}
  
We now briefly analyze the running times of algorithms on all parties.  The running time of computing $P'_i$ on each machine $i$ in round 1 is $O(n_i^2)$ and this is the bottleneck for machine $i$. Considering all possible values of $L$, the running time on machine $i$ is $O\left(n_i^2\cdot \frac{\log \Delta}{\epsilon} \right)$. The running time of the round-4 algorithm of the central coordinator for one $L$ will be $O\left(\left(\frac{km}{\epsilon}\right)^2\right)$. We sort all the interesting $L$ values in increasing order. The trick here is to run only \emph{the first two} rounds of the main algorithm. The central coordinator then use binary search to find the smallest $L'$ that the main algorithm sends out ``Yes'' in Round 2. And it proceeds to Round 3 and 4 only for this single $L'$. So, the running time of the central coordinator can be made $O\left(\left(\frac{km}{\epsilon}\right)^2\right)$.

%\cdot \frac{\log \Delta}{\epsilon}
The quadratic dependence of running time of machine $i$ on $n_i$ might be an issue when $n_i$ is big; we discuss how to alleviate the issue in the supplementary material. 
	\section{Conclusion}
In this paper, we give a distributed $(24(1+\epsilon), 1+\epsilon)$-bicriteria approximation for the $(k, z)$-center problem, with communication cost $O\left(\frac{km}{\epsilon}+\frac{m\log \Delta}{\epsilon}\right)$. The running times of the algorithms for all parties are polynomial. We evaluate the algorithm on realworld data sets and it outperforms most previous algorithms, matching the performance of the state-of-art method\cite{DBLP:conf/spaa/GuhaLZ17}.
%TODO: 
 %We remark that the number of points sent between the parties in our main algorithm $\distkzcmulti$ is at most $km(1/\epsilon+1)$ (there is no hidden constant). Thus, in cases where $L^*$ is known or falls in a small known interval, our algorithm has much smaller communication cost than that of the $(2+\epsilon, O(1))$-bicriteria approximation of Guha et al.\ \cite{DBLP:conf/spaa/GuhaLZ17} which involves a$O(\frac{\log z}{\epsilon})$).

For the $(k, z)$-median/means problem, we give a distributed $(1+\epsilon, 1+\epsilon)$-bicriteria approximation algorithm with communication cost $O\left(\Phi D\cdot\frac{\log \Delta}{\epsilon}\right)$, where $\Phi$ is the upper bound on the size of the coreset constructed using the algorithm of \cite{DBLP:conf/nips/BalcanEL13}.  The central coordinator needs to solve the optimization problem of finding a solution that is simultaneously good for $O\left(\frac{\log (\Delta n/\epsilon)}{\epsilon}\right)$ $k$-median/means instances.  Since the approximation ratio for this problem will go to \emph{both} factors in the bicriteria ratio, we really need a $(1+\epsilon)$-approximation for the optimization problem. Unfortunately, solving $k$-median/means alone is already APX-hard, and we don't know a heuristic algorithm that works well in practice (e.g, a counterpart to Lloyd's algorithm for $k$-means). It is interesting to study if a different approach can lead to a polynomial time distributed algorithm with $O(1)$-approximation guarantee. 
	
	\bibliographystyle{plain}
	\bibliography{distributed-clustering-with-outliers}

\begin{thebibliography}{10}

\bibitem{ANSW16}
Sara Ahmadian, Ashkan Norouzi{-}Fard, Ola Svensson, and Justin Ward.
\newblock Better guarantees for k-means and euclidean k-median by primal-dual
  algorithms.
\newblock In {\em 58th {IEEE} Annual Symposium on Foundations of Computer
  Science, {FOCS} 2017, Berkeley, CA, USA, October 15-17, 2017}, pages 61--72,
  2017.

\bibitem{DBLP:journals/mor/AnthonyGGN10}
Barbara~M. Anthony, Vineet Goyal, Anupam Gupta, and Viswanath Nagarajan.
\newblock A plant location guide for the unsure: Approximation algorithms for
  min-max location problems.
\newblock {\em Math. Oper. Res.}, 35(1):79--101, 2010.

\bibitem{DBLP:conf/nips/BalcanEL13}
Maria{-}Florina Balcan, Steven Ehrlich, and Yingyu Liang.
\newblock Distributed k-means and k-median clustering on general communication
  topologies.
\newblock In {\em Advances in Neural Information Processing Systems 26, NIPS
  2013, December 5-8, 2013, Lake Tahoe, Nevada, United States.}, pages
  1995--2003, 2013.

\bibitem{BPRST17}
Jaroslaw Byrka, Thomas Pensyl, Bartosz Rybicki, Aravind Srinivasan, and Khoa
  Trinh.
\newblock An improved approximation for \emph{k}-median and positive
  correlation in budgeted optimization.
\newblock {\em {ACM} Trans. Algorithms}, 13(2):23:1--23:31, 2017.

\bibitem{DBLP:conf/soda/CharikarKMN01}
Moses Charikar, Samir Khuller, David~M. Mount, and Giri Narasimhan.
\newblock Algorithms for facility location problems with outliers.
\newblock In {\em Proceedings of the 12th Annual Symposium on Discrete
  Algorithms, January 7-9, 2001, Washington, DC, {USA.}}, pages 642--651, 2001.

\bibitem{DBLP:conf/sdm/ChawlaG13}
Sanjay Chawla and Aristides Gionis.
\newblock k-means$--$: {A} unified approach to clustering and outlier
  detection.
\newblock In {\em Proceedings of the 13th {SIAM} International Conference on
  Data Mining, May 2-4, 2013. Austin, Texas, {USA.}}, pages 189--197, 2013.

\bibitem{DBLP:journals/corr/abs-1805-09495}
Jiecao Chen, Erfan~Sadeqi Azer, and Qin Zhang.
\newblock A practical algorithm for distributed clustering and outlier
  detection.
\newblock {\em CoRR}, abs/1805.09495, 2018.

\bibitem{DBLP:conf/nips/ChenSWZ16}
Jiecao Chen, He~Sun, David~P. Woodruff, and Qin Zhang.
\newblock Communication-optimal distributed clustering.
\newblock In {\em Advances in Neural Information Processing Systems 29: Annual
  Conference on Neural Information Processing Systems 2016, December 5-10,
  2016, Barcelona, Spain}, pages 3720--3728, 2016.

\bibitem{DBLP:conf/soda/Chen08}
Ke~Chen.
\newblock A constant factor approximation algorithm for \emph{k}-median
  clustering with outliers.
\newblock In {\em Proceedings of the 19th Annual {ACM-SIAM} Symposium on
  Discrete Algorithms, {SODA} 2008, San Francisco, California, USA, January
  20-22, 2008}, pages 826--835, 2008.

\bibitem{DBLP:conf/icml/DingLHL16}
Hu~Ding, Yu~Liu, Lingxiao Huang, and Jian Li.
\newblock $k$-means clustering with distributed dimensions.
\newblock In {\em Proceedings of the 33rd International Conference on Machine
  Learning, {ICML} 2016, New York City, NY, USA, June 19-24, 2016}, pages
  1339--1348, 2016.

\bibitem{DBLP:conf/kdd/EneIM11}
Alina Ene, Sungjin Im, and Benjamin Moseley.
\newblock Fast clustering using mapreduce.
\newblock In {\em Proceedings of the 17th {ACM} {SIGKDD} International
  Conference on Knowledge Discovery and Data Mining, San Diego, CA, USA, August
  21-24, 2011}, pages 681--689, 2011.

\bibitem{DBLP:conf/spaa/GuhaLZ17}
Sudipto Guha, Yi~Li, and Qin Zhang.
\newblock Distributed partial clustering.
\newblock In {\em Proceedings of the 29th {ACM} Symposium on Parallelism in
  Algorithms and Architectures, {SPAA} 2017, Washington DC, USA, July 24-26,
  2017}, pages 143--152, 2017.

\bibitem{DBLP:journals/tkde/GuhaMMMO03}
Sudipto Guha, Adam Meyerson, Nina Mishra, Rajeev Motwani, and Liadan
  O'Callaghan.
\newblock Clustering data streams: Theory and practice.
\newblock {\em {IEEE} Trans. Knowl. Data Eng.}, 15(3):515--528, 2003.

\bibitem{DBLP:journals/toc/HastadW07}
Johan H{\aa}stad and Avi Wigderson.
\newblock The randomized communication complexity of set disjointness.
\newblock {\em Theory of Computing}, 3(1):211--219, 2007.

\bibitem{DBLP:journals/mor/HochbaumS85}
Dorit~S. Hochbaum and David~B. Shmoys.
\newblock A best possible heuristic for the \emph{k}-center problem.
\newblock {\em Math. Oper. Ues.}, 10(2):180--184, 1985.

\bibitem{DBLP:conf/spaa/ImM15}
Sungjin Im and Benjamin Moseley.
\newblock Brief announcement: Fast and better distributed mapreduce algorithms
  for k-center clustering.
\newblock In {\em Proceedings of the 27th {ACM} on Symposium on Parallelism in
  Algorithms and Architectures, {SPAA} 2015, Portland, OR, USA, June 13-15,
  2015}, pages 65--67, 2015.

\bibitem{KLS18}
Ravishankar Krishnaswamy, Shi Li, and Sai Sandeep.
\newblock Constant approximation for k-median and k-means with outliers via
  iterative rounding.
\newblock In {\em Proceedings of the 50th Annual {ACM} {SIGACT} Symposium on
  Theory of Computing, {STOC} 2018, Los Angeles, CA, USA, June 25-29, 2018},
  pages 646--659, 2018.

\bibitem{Lichman:2013}
M.~Lichman.
\newblock {UCI} machine learning repository, 2013.

\bibitem{Lloyd06}
Stuart~P. Lloyd.
\newblock Least squares quantization in {PCM}.
\newblock {\em {IEEE} Trans. Information Theory}, 28(2):129--136, 1982.

\bibitem{DBLP:conf/nips/MalkomesKCWM15}
Gustavo Malkomes, Matt~J. Kusner, Wenlin Chen, Kilian~Q. Weinberger, and
  Benjamin Moseley.
\newblock Fast distributed k-center clustering with outliers on massive data.
\newblock In {\em Advances in Neural Information Processing Systems 28, NIPS
  2015, December 7-12, 2015, Montreal, Quebec, Canada}, pages 1063--1071, 2015.

\bibitem{DBLP:conf/icassp/tsanaslmr10}
Athanasios Tsanas, Max~A. Little, Patrick~E. Mcsharry, and Lorraine~O. Ramig.
\newblock Enhanced classical dysphonia measures and sparse regression for
  telemonitoring of parkinson's disease progression.
\newblock In {\em Proceedings of the {IEEE} International Conference on
  Acoustics, Speech, and Signal Processing, {ICASSP} 2010, 14-19 March 2010,
  Dallas, Texas, {USA}}, pages 594--597, 2010.

\end{thebibliography}
	
	\appendix
	\section{Necessity of Linear Dependence of Communication Cost on $z$ for True Approximation Algorithms}
In this section, we show that if one is aiming for a multiplicative approximation for the $(k, z)$-center, $(k, z)$-median, or $(k,z)$-means problem, then the communication cost is at least $\Omega(z)$ bits, even if there are only 2 machines.   We show that deciding whether the optimum $(k, z)$-center solution has cost $0$ or not requires $\Omega(z)$ bits of communication.  This holds for any combination of values for $n, k$ and $z$ such that $k + z \leq n-1$. Let $B = 1$. The points are all in the real line $\R$. On machine 1, there are $n - z - 2$ copies of points from the set $\set{-1, -2, \cdots, -(k-1)}$, where each one of the $k-1$ points appears either $\floor{\frac{n-z - 2}{k-1}}$ or $\ceil{\frac{n-z - 2}{k-1}}$ times. Notice that each point in the set appears at least once in the set. Meanwhile, machine 1 has a set $A$ of different points in $[2(z+2)]$, and machine 2 has a set $B$ of different points in $[2(z+2)]$, and we have $|A| + |B| = z + 2$.  If $A \cap B \neq \emptyset$, then the cost of the optimum solution is 0. Let  $e \in A \cap B$, then we can discard all points except $e$ from $A$ and $B$. Then we discarded exactly $z$ points and the remaining set of points are at $k-1 + 1 = k$ locations.  On the other hand, if $A \cap B = \emptyset$, then the cost of the optimum solution is not 0. Thus deciding whether the cost is 0 or not requires us to decide if $A \cap B = \emptyset$, which is exactly the \emph{set disjointness} problem. This requires a communication cost of $\Omega(z)$ between machine 1 and machine 2\footnote{This is a well-known result in communication complexity theory, see e.g. \cite{DBLP:journals/toc/HastadW07}}.

\section{Dealing with Various Issues of the Algorithm for $(k, z)$-Center}
In this section, we show how to handle various issues that our $(k, z)$-center algorithm might face.

\textbf{When $d_{\min}$ and $d_{\max}$ are not given.} We can remove the assumption that $d_{\min}$ and $d_{\max}$ are given to us.  Let $d_{\min, i}$ and $d_{\max, i}$ be the minimum and maximum non-zero pairwise distances between points in $P_i$. The crucial observation is that running \textsf{aggregating} on $P_i$ for $L < d_{\min, i}$ is the same as running it for $L = 0$, and running it for $L > d_{\max, i}$ is the same as running it for $L = d_{\max, i}$.  Thus, machine $i$ only needs to consider $L$ values that are integer powers of $1+\epsilon$ inside $[d_{\min, i}, (1+\epsilon) d_{\max, i})$, or $0$, and send results for these $L$ values. Since $d_{\min, i} \geq d_{\min}$ and $d_{\max, i}\geq l_{\max}$, the number of such $L$ values is at most $O\left(\frac{\log \Delta}{\epsilon}\right)$.  Also notice that the data points sent from machine $i$ to the coordinator are all generated from $P_i$.  Thus, the aspect ratio for the union of all points received by the coordinator, is at most $\Delta$. This can guarantee that the coordinator only needs to use $O(\log \frac{\log \Delta}{\epsilon})$ iterations in the binary search step in Round 4.

\textbf{When $\Delta$ is super big.} There are many ways to handle the case when $\Delta$ is super-large. In many applications, we know the nature of the dataset and have a reasonable guess on $L^*$.  In other applications, we may be only interested in the case where $L^* \in [A, B]$: we are happy with any clustering of cost less than $A$ and any clustering of cost more than $B$ is meaningless.  In these applications where we have inside information about the dataset, the number of guesses can be much smaller. Finally, if we allow more rounds in our algorithm, we can use binary search for the whole algorithm $\distkzcmulti$, not just inside Round 4. We only need to run the algorithm for $O\left(\log\frac{\log\Delta}{\epsilon}\right)$ iterations; this will increase the number of rounds to $O\left(\log\frac{\log\Delta}{\epsilon}\right)$, but decrease the communication cost to $O\left(\frac{km}{\epsilon}+m\log\frac{\log \Delta}{\epsilon}\right)$.

\textbf{Handling the Quadratic Running Time of Round 1 on Machine $i$.} In Round 1 of the algorithm $\distkzcmulti$, each machine $i$ needs to run $\mathsf{aggregating}$ on $n_i = |P_i|$ points, leading to a running time of order $O(n_i^2)$.  In cases where $n_i$ is large, the algorithm might be slow. %We give some approaches one might try to improve the running time.
%
%If the dimension $B$ is small (say, a small constant), we can maintain a k-B tree data structure for points in $P_i$. Using that we can count the number of points in a ball efficiently.  By allowing $(1+\epsilon)$-distortion on the query balls,  we can make sure that each ball contains $(1/\epsilon)^B$ regions from the k-B tree. This may improve the running time. %TODO: https://www.cs.umd.edu/~mount/Papers/cgta00-range.pdf or some other paper.  In case the dimension $B$ is big, we may use locally sensitive hashing (LSH) to reduce the running time.  By losing a factor of $\rho \geq 1$ in the approximation ratio, we can reduce the dependence of the running time on $n_i$ to $n_i^{1+1/\rho}$. 
%
We can decrease the running time, at the price of increasing the communication cost and the running time on the coordinator. We view each $i \in [m]$ as a collection of $t_i \geq 1$ sub-machines, for some integer $t_i  \in [1, n_i]$. Then, we run $\distkzcmulti$ on the set of $\sum_{i \in [m]}t_i$ sub-machines, instead of the original set of $m$ machines. The communication cost of the algorithm $\distkzcmulti$ increases to $O\left(\frac{k\sum_{i \in [m]} t_i}{\epsilon}\cdot \frac{\log \Delta}{\epsilon}\right)$, and the running time on each machine $i$ decreases to $O\left(\frac {n_i}{t_i})^2\cdot t_i\cdot \frac{\log \Delta}{\epsilon}\right) = O\left(\frac{n_i^2}{t_i}\cdot\frac{\log \Delta}{\epsilon}\right)$, and the running time of the algorithm for the coordinator becomes $O\left(\left(\frac{k\sum_{i\in[m]}t_i}{\epsilon}\right)^2\cdot\log\frac{\log \Delta}{\epsilon}\right)$.  Each machine $i$ can choose a $t_i$ so that the $O\left(\frac{n_i^2}{t_i}\cdot\frac{\log \Delta}{\epsilon}\right)$-time algorithm of Round 1 terminates in acceptable amount of time. %, and send it to the coordinator in Round 1.

	%No need to kow d_min and d_max for k-median/means.
\section{Distributed Algorithms $(k,z)$-Median/Means}\label{sec:app-kzmedian}
%Given a metric $d'$ over the universe $\R^D$, a set $Q \subseteq \R^D$ of with weights $w \in \R_{\geq 0}^Q$, and any set $C$ of $k$ centers , we define
%\begin{align*}
%	\kmcost_{d', (Q, w)}(C) := \sum_{q \in Q}w_q d'(q, C)
%\end{align*}
%to be the cost of the $k$-median solution $C$ to the instance defined by $d'$ and $(Q, w)$. We simply use $\kmcost(d', Q)$ for $\kmcost(d', (Q, \vec 1))$, where $\vec 1$ is the vector of 1's.

In this section, we give our distributed algorithm for the $(k, z)$-median/means problems in Euclidean metrics.  Let $m, k, z, \epsilon, n, P \subseteq \R^D$ and $\set{P_i}_{i \in [m]}$ be as defined in the problem setting.  Let  $\delta > 0$ be the confidence parameter; i.e, our algorithm needs to succeed with probability $1-\delta$.  Also,  we define a parameter $\ell \in \{1, 2\}$ to indicate whether the problem we are considering is $(k, z)$-median ($\ell = 1$) or $(k, z)$-means ($\ell=2$). 

Recall that $d_{\min}$ and $d_{\max}$ are respectively the minimum and maximum non-zero pairwise distance between points in $P$. It is not hard to see that the optimum solution to the instance has cost either $0$ or at least $d^\ell_{\min}/\ell$. For a technical reason, we can redefine $d(p, q)$ as follows for every $p, q \in \R^D$: 
\begin{align*}
	d(p, q) = \begin{cases}
		0 & \text{ if } \|p -q\|_2 = 0 \\
		\min\Big\{\max\{\|p-q\|, \epsilon d_{\min}/(2n)\}, 2d_{\max} \Big\}& \text{ otherwise}
	\end{cases}.
\end{align*}
That is, we truncate distances below at $\epsilon d_{\min}/(2n)$, and above at $2d_{\max}$. It is easy to see that the problem w.r.t the new metric is equivalent to the original one up to a multiplicative factor of $1+\epsilon$. In the new instance, we have either $d(p, q) = 0$ or $d(p, q) \in [\epsilon d_{\min}/(2n), 2d_{\max}]$.

Given an integer $z' \in [0, n)$ and a set $C$ of $k$ centers,  we define 
\begin{align*}
	\cost_{z'}(C) := \min_{P' \subseteq P:|P'| = n - z'} \sum_{p \in P'} d^\ell(p, C)
\end{align*}
to be the cost of the solution $C$ to the $(k, z)$-median/mean instance defined by $P, d$ and $z'$.  In the above definition, we remove $z'$ outliers and consider the cost incurred by the $n-z'$ non-outliers.  Notice the set $P'$ that minimizes the cost is  the set of $n-z'$ points in $P$ that are closest to $C$.  

For some technical reason, we need to allow $z'$ to take real values in $[0, n)$. In this case, we define
\begin{align*}
	\cost_{z'}(C) := \min_{w' \in [0, 1]^P:w'(P) = n-z'} \sum_{p \in P} w'_pd^\ell(p, C).
\end{align*}
Given a set $C$ of $k$ centers, the optimum $w'$ can be obtained in a greedy manner: assign 1 to the $n-\ceil{z'}$ points in $P$ that are closest to $C$, assign $\ceil{z'} - z'$ to the point in $P$ that is the $n-\ceil{z'} + 1$-th closest to $C$, and assign $0$ to the remaining points.

\subsection{The $(k,z)$-Median/Means Problem Reformulated}  %TODO: think about the use of \cost, now it is confusing
In this section, we reformulate the $(k, z)$-median/means problems in a way that will be useful for our algorithm design.  
%TODO: define $\cost_{P, d, z}$: For a clustering problem on dataset $P$ with $k$ centers $C\subseteq C^\circ$, define the cost with a threshold distance $d_L$ as  $\sum_{p \in P}d_L(p,C) = \sum_{p\in P}d_L(p, C)$,
%
%In this section, we give a new formulation of the $(k, z)$-median problem. %In the new formulation, we are given a family of $k$-median/means instances (without outliers), and need to find a solution that is simultaneously good all the instances in the family.   
%
Given a threshold $L \geq 0$, we define
 $d_L(p, q)=\min\{d(p, q), L\}$ for every two points $p, q \in \R^D$. In other words, $d_L$ is the metric $d$ with distances truncated at $L$.  The following crucial lemma gives the reformulations of $k$-median/means problems:
\begin{mylmm} \label{lemma:reformulate}
	For any real number $z' \in [0, n)$, and any set $C$ of $k$ centers, we have 
	\begin{equation}
	  \cost_{z'}(C) = \sup_{L\geq 0}\big(\sum_{p \in P}d^\ell_L(p, C) - z'L^{\ell}\big).  \label{eq:cost-outlier}
	\end{equation}
	Moreover, the superior is achieved when $L$ is the $(n-\floor{z'})$-th smallest number in the multi-set $\set{d(p, C):p \in P}$.
\end{mylmm}

\begin{proof}
  Let $\bar L$ be the $(n-\floor{z'})$-th smallest number in the multi-set $\set{d(p, C):p \in P}$. Then it can be seen that $\cost_{z'}(C) = \sum_{p \in P}d^\ell_{\bar L}(p, C) - z'{\bar L}^\ell$.  Indeed, $\cost_{z'}(C)$ is the sum of the $n-z'$ smallest numbers in $S:=\{d^\ell(p, C): p \in P\}$. (When $n-z'$ is not an integer, then we take a fraction of the last number.) To compute the quantity on the right side, we truncate the numbers in $S$ at $\bar L^\ell$, and then take the sum of the truncated numbers minus $z'\bar L^\ell$. Since $\bar L^\ell$ is the $(n - \floor{z'})$-th smallest number in $S$, this quantity is exactly $\cost_{z'}(C)$.

  It remains to prove that $\sum_{p \in P}d^\ell_L(p,C) - z'L^\ell$ attains its maximum value at $L = {\bar L}$. First consider any $L<{\bar L}$, and define $P'=\set{p\in P| L< d(p,C)<{\bar L}}$, and $P'' = \set{p \in P: d(p, C) \geq \bar L}$. By the definition of $\bar L$, we have $|P''| \geq \floor{z'} + 1 > z'$. Then, we have
  \begin{align*}
    &\left(\sum_{p \in P}d^\ell_{\bar L}(p,C) - z'{\bar L}^\ell\right) - \left(\sum_{p \in P}d^\ell_L(p,C) - z'L^\ell\right) \\
    &=\sum_{p \in P'} (d^\ell(p, C) - L^\ell) + |P''|(\bar L^\ell - L^\ell)  - z'(\bar L^\ell - L^\ell) \geq \sum_{p\in P'}(d^\ell(p, C)-L^\ell) \geq 0.
  \end{align*}
  Now consider any $L>{\bar L}$ and define $P'=\set{p\in P: \bar L < d(p,C)< L}$ and $P'' = \set{p \in P: d(p, C) \geq L}$. By the definition of $\bar L$, we have  $|P'\cup P''| = \big|\set{p \in P: d(p, C) > \bar L}\big| \leq \floor{z'} \leq z'$. Then, we have
  \begin{align*}
    &\left(\sum_{p \in P}d_{\bar L^\ell}(p,C) - z'{\bar L}^\ell\right) - \left(\sum_{p \in P}d^\ell_L(p,C) - z'L^\ell\right) \\
   &=-\sum_{p \in P'} (d^\ell(p, C) - \bar L ^\ell ) - |P''|(L^\ell - \bar L^\ell) + z'(L^\ell - \bar L^\ell) \\
    &\geq -|P'|(L^\ell - \bar L^\ell) - |P''|(L^\ell - \bar L^\ell)  + z' (L^\ell - \bar L^\ell) \geq 0.
  \end{align*}
  This finishes the proof of the lemma. %TODO: figure for the proof..
\end{proof}

With the above lemma, the $(k,z)$-median/means problem becomes finding a set of $k$ centers $C \subseteq \R^D$ so as to minimize $\sup_{L\geq 0}\,(\sum_{p \in P}d^\ell_L(p,C) - zL^\ell)$.  %That is, we need find a solution that is simultaneously good for all solve the $k$-median instances $(P, d_L)$. 
To get a handle on the problem, we first discretize the value space for $L$.  Formally, we only allow $L$ to take values in
\begin{align*}
	\bbL := \{0\} \cup \Big(\set{(1+\epsilon)^t: t \in \Z} \cap  (\epsilon d_{\min}/(2(1+\epsilon)n), 2d_{\max}] \Big).
\end{align*}
Then, we have $|\bbL| = O\left(\frac{\log (\Delta n/\epsilon)}{\epsilon}\right)$.  We define $\cost'_{z'}(C)$ as in \eqref{eq:cost-outlier}, except that we only consider $L$ values in $\bbL$. That is, for every $z' \in [0, n)$ and a set $C$ of $k$ centers, we define
\begin{equation}\label{eq:rounded-prob}
 \cost'_{z'}(C):=\sup_{L\in \bbL}\left(\sum_{p \in P}d^\ell_L(p,C) - z'L^\ell\right).
\end{equation}

%And define the cost in the rounded problem as
%\begin{equation}\label{eq:rounded-prob-cost}
%  \cost'(C) = \sup_{L\in \bbL}\,(\sum_{p \in P}d_L(p,C) - zL)
%\end{equation}
For a fixed  $z'$ and $C$, we have $\cost'_{z'}(C) \leq \cost_{z'}(C)$,  since the supreme is taken over a subset of $L$ values in the definition of $\cost'_{z'}(C)$.  %To prove the other direction, we need to make the following definition:
%\begin{mydef} %TODO: mention inliers in problem definition.
%	A set $C\subseteq \R^D$ of $k$ centers w.r.t removing $z'$ outliers is said to be \emph{reasonable} if 
%	\begin{enumerate}
%		\item $\max_{p \in P, c \in C}d(p, C)\leq 2d_{\max}$, 
%		\item $\cost_{z'}(C) = 0$ or $\cost_{z'}(C) \geq d^\ell_{\min}/2$.
%	\end{enumerate}
%\end{mydef}
%
%Obviously, we can require the set $C$ of centers to be reasonable in the definition of the $k$-median/means problems: if $d(p, c) > 2d_{\max}$ for some $c \in C$ and $p \in P$, then $d(p', c) > d_{\max}$ for every $p' \in P$ by triangle inequalities and the definition of $d_{\max}$. Thus we can simply replace $c$ with any point in $P$.  If some center $c$ is connected by $c$ outliers
%
%
%Such a $C$ is worse than choosing an arbitrary set of $k$ points in $P$. Consider the other case where for all $p \in I$ we have $d(p, C) < d_{\min}/2)$ for some $p$. By the definition of $d_{\min}$, each of the $k$ clusters of inliers must be collocated. In this case, the optimum solution has cost $0$.
%
Now we show the other direction of the inequality:
\begin{mylmm}\label{lmm:R-to-bbL}
	For every set $C$ of $k$ centers, and any $z' \in [0, n]$,  we have 
  \begin{equation}
	 \cost_{(1+\epsilon)^\ell z'}(C) \leq (1+\epsilon)^\ell \cost'_{z'}(C).
  \end{equation}
\end{mylmm}

\begin{proof}
	By Lemma~\ref{lemma:reformulate}, we have that $\cost_{(1+\epsilon)^\ell {z'}}(C) = \sup_{L\geq 0} \left(\sum_{p \in P}d^\ell_L(p,C) - (1+\epsilon)^\ell {z'}L^\ell\right)$.  Let $\bar L$ be the $L \in \R$ that achieves the maximum value. %by Lemma~\ref{lemma:reformulate}, $\bar L$ is the $(n-\floor{(1+\epsilon)^\ell z'})$-th smallest value in $\{d(p, C):p \in P\}$. 
	Thus, $\cost_{(1+\epsilon)^\ell {z'}}(C) = \sum_{p \in P}d^\ell_{\bar L}(p,C) - (1+\epsilon)^\ell {z'}\bar L^\ell$. 
	By Lemma~\ref{lemma:reformulate} and the new definition of the metric $d$, we have $\bar L = 0$ or $\bar L \in [\epsilon d_{\min}/(2n), 2d_{\max}]$. Thus there is always a $L' \in \bbL$ such that $\bar L \in [L', (1+\epsilon)L')$. 
	\begin{align*}
	   \cost_{(1+\epsilon)^\ell {z'}}(C) &= \sum_{p \in P}d^\ell_{\bar L}(p,C) - (1+\epsilon)^\ell {z'}\bar L^\ell\\
	    &\leq(1+\epsilon)^\ell\sum_{p \in P}d^\ell_{L'}(p,C)-(1+\epsilon)^\ell {z'}L'^\ell\leq (1+\epsilon)^\ell\cost'_{{z'}}(C).
	 \end{align*}
	 The first inequality is by $L' \leq \bar L < (1+\epsilon)L'$ and the second inequality is by the definition of $\cost'_{z'}(C)$ and the fact that $L' \in \bbL$.	 
\end{proof}
The lemma allows us to focus on the new objective function $\cost'_{\tilde z}(C)$ for some suitably defined $\tilde z$.

%In the following we will focus on the rounded problem \eqref{eq:rounded-prob} under the distributed setting:
%\begin{quote}
%	Suppose a huge dataset $P=\{p_1,p_2,\ldots,p_n\}$ is partitioned into $m$ different parts $P_1,P_2,\ldots,P_m$, such that $P_i$ is stored on machine $i$, $\bigcup_{i=1}^mP_i=P$, and $\forall i\neq j, P_i\cap P_j = \emptyset$. We want to find a set of $k$ centers (denote by $C$), so as to minimize $\cost'_z(C)$ with small communication cost.
%\end{quote}

%Let $C^*$ denote the optimal solution of the original problem \eqref{eq:orig-prob}, $C'^*$ denote the optimal solution of the rounded problem \eqref{eq:rounded-prob}. Then $C'^*$ is a $(1+O(\epsilon))$-approximate solution of the original problem \eqref{eq:orig-prob}, i.e.,

\subsection{Distributed Algorithm for the Reformulated Problem via $\epsilon$-Coresets}
%Our approach is based on the coreset technique\cite{DBLP:conf/stoc/FeldmanL11, DBLP:conf/soda/FeldmanS12, DBLP:conf/soda/FeldmanSS13, DBLP:conf/compgeom/FeldmanFSS07}

An important notion that has been used to design efficient algorithms for $k$-median/means in Euclidean space is the $\epsilon$-coreset. Roughly speaking, it is a weighted set of points that approximates the given set $P$ well.  Formally, 
\begin{mydef}%[$\epsilon$-coreset\cite{DBLP:conf/stoc/FeldmanL11}]
 A weighted set $(Q,w)$ of points is an $\epsilon$-coreset for $P'$ w.r.t. distance $d'$, if for every set $C \subseteq \R^D$ of $k$ centers, we have
  \begin{equation*}
    \left(\sum_{q\in Q}w_qd'^\ell(q, C)\right) \Big/ \left(\sum_{p\in P'}d'^\ell(p, C)\right)\in [1-\epsilon, 1 + \epsilon].
  \end{equation*}
\end{mydef}

The following theorem from \cite{DBLP:conf/nips/BalcanEL13} gives a distributed algorithm to construct $\epsilon$-coresets for the points $P$ and a truncated metric $d_L$:
\begin{mythm} \cite{DBLP:conf/nips/BalcanEL13}
	\label{thm:coreset}
	Given $\delta > 0, \epsilon > 0, L \geq 0$, there is an $2$-round distributed algorithm that outputs an $\epsilon$-coreset $(Q, w)$ of $P$ w.r.t distance $d^L$, with probability at least $1-\delta$. The size of the coreset is at most $\Phi$, where $\Phi = O\left(\frac{1}{\epsilon^2}(kD + \log\frac1\delta) + mk\right)$ for $k$-median, and $\Phi = O\left(\frac{1}{\epsilon^4}(kD + \log\frac1\delta)+mk\log\frac{mk}{\delta}\right)$ for $k$-means. The communication complexity of the algorithm is $O(D\Phi)$.
\end{mythm}
The correspondent theorem in \cite{DBLP:conf/nips/BalcanEL13} only considers the original Euclidean metric $\|\cdot-\cdot\|_2$. In our definition of $d_L$, we truncated distances below at $\epsilon \cdot d_{\min}/(2n)$, and then above at $L$.  But it is easy to extend their theorem so that it works for the truncated metrics, since all we need is that the metric has $O(D)$ ``pseudo-dimension'' (defined in \cite{DBLP:conf/nips/BalcanEL13}). Truncating the metric only change the pseudo-dimension by an additive constant.  From now on, let $\Phi$ be the upper bound on the size of the $\epsilon$-coreset in Theorem~\ref{thm:coreset}.

%The advantage of coreset is that, if $Q_i$ is an $\epsilon$-coreset of $P_i$ then $Q=\bigcup_iQ_i$ is also an $\epsilon$-coreset of $\bigcup_iP_i=P$. %So for a given $L$, we can sample a $Q_i$ on each machine $i$ w.r.t. $d_L$, and send those $Q_i$'s to a single machine $S$; then on $S$ we have a $\epsilon$-coreset for the whole dataset $P$ w.r.t. $d_L$.

%\begin{algorithm}[htbp]
%  \caption{Local Clustering with Outliers\label{alg:local}}
%  \begin{small}
%    \begin{algorithmic}[1]
%      \textbf{input}: {$N$ different coresets with corresponding threshold distance ${(Q^{L_i}, d_{L_i})}_{i=1}^N$, $\epsilon$, $z$}\\
%      \textbf{output}: {set of $k$ centers ${\tilde C}$ that approximately minimize $\sup_{L\in \bbL}\, (\sum_{q \in Q_L}d_L(q, {\tilde C}) - (1+\epsilon)zL)$ outliers}
%      \STATE 
%%      \xyguoRFC{Don't know how to do, just assume it can be done efficiently.}\;
%    \end{algorithmic}
%  \end{small}
%\end{algorithm}

%TODO: \tilde O hide logarithmic factors in n, ...$

With Theorem~\ref{thm:coreset} in hand, it is straightforward to give our algorithm for $(k, z)$-median/means. %described in described in Algorithm~\ref{alg:kzmedian}.
For all $L \in \bbL$, we run in parallel the 2-round distributed algorithm in Theorem~\ref{thm:coreset} with $\delta$ scaled down by a factor of $|\bbL|$ to obtain a $\epsilon$-coreset $(Q_L, w_L)$.   The communication cost of the algorithm is then $\Phi D \cdot \frac{\log(n\Delta/\epsilon)}{\epsilon}$. 

Let $\tilde z = \frac{(1+\epsilon)^2z}{1-\epsilon}$.  We would like to find a set $\tilde C$ of $k$ points that minimizes $\sup_{L \in \bbL}\left(\frac{1}{1-\epsilon}\sum_{q \in Q_L}w_qd^\ell_L(q, \tilde C) - \tilde zL^\ell\right)$. However, it is not even clear whether the optimum $\tilde C$ can be represented using finite number of bits or not. Instead, the coordinator will output a set $\tilde C \subseteq \R^D$ of $k$ centers such that for every set $C^*\subseteq \R^D$ of $k$ centers, we have 
\begin{align}
	\sup_{L\in \bbL}\left(\frac{1}{1-\epsilon}\sum_{q \in Q_L}w_qd^\ell_L(q, {\tilde C}) - \tilde zL^\ell\right) \leq \sup_{L\in \bbL}\left(\frac{1+\epsilon}{1-\epsilon}\sum_{q \in Q_L}w_qd^\ell_L(q, C^*) - \tilde zL^\ell\right). \label{inequ:tilde-C-is-approximate}
\end{align}
The extra $(1+\epsilon)$ factor on the right-side allows us to partition the Euclidean space into finite number of cells.  
This can be done by partitioning the space into $O\left(\frac{\log (\Delta n/\epsilon)}\epsilon\right)^{|\bbL|\Phi}$ cells so that all points in a cell have similar respective distances to all points in $\union_{L \in \bbL}Q_L$. So, we can choose an arbitrary representative point from each cell, and then enumerate all sets $\tilde C$ of $k$ representatives and output the one with the minimum $\sup_{L\in \bbL}\left(\frac{1}{1-\epsilon}\sum_{q \in Q_L}w_qd^\ell_L(q, {\tilde C}) - \tilde zL^\ell\right)$. The running time of the algorithm can be bounded by $\exp\left(\Phi, k, |\bbL|, D, \log\left(\frac{\log (n\Delta/\epsilon)}{\epsilon}\right)\right) = \exp\left(\poly\left(\frac1\epsilon, k, D, m, \log \frac1\delta, \log \Delta\right)\right)$.

%\begin{algorithm}[htbp]
%  \caption{Distributed Clustering with Outliers}\label{alg:kzmedian}
%      \textbf{input on all parties}: {$m, k, z, \epsilon \in (0, 1)$, and $\delta \in (0, 1)$}\\
%      \textbf{input on machine $i$}: set $P_i$ of points\\
%      \textbf{output}: {set of $k$ centers ${\tilde C}$}
%  \begin{small}
%  \linebelow
%  
%\textbf{Round 1 on  machine $i \in [m]$}    
%
%	\begin{algorithmic}[1]
%      \STATE \textbf{for} $L\in\bbL$ \textbf{do}
%      \STATE\hspace*{\algorithmicindent} construct a random $O\left(\frac{1}{\epsilon^2}(kD+\log\frac{1}{\delta})+mk\right)$-size weighted set $(Q_L_i, w_L_i)$ %using theorem \cite{DBLP:conf/nips/BalcanEL13};
%      such that w.p at least $1 - \delta/m$, $(Q_L_i, w_L_i)$ is an $\epsilon$-coreset of $P_i$ w.r.t. $d_L$ 
%      \STATE\hspace*{\algorithmicindent} send $Q_L_i$ to the central coordinator
%    %  \STATE \COMMENT{On server $S$}
%      \end{algorithmic}
%  \linebelow
%  
%\textbf{Round 2 on the central coordinator}    
%
%	\begin{algorithmic}[1]
%      \STATE \textbf{for} every $L \in \bbL$ \textbf{do}: let $Q_L \gets \bigcup_iQ_L_i$ and $w_L \in \R^{Q_L}$ be the concatenation of $\set{w_L_i}_{i \in [m]}$
%      \STATE 
%      	
%      \end{algorithmic}
%    \end{small}
%\end{algorithm}

%TODO: $C^*$
\subsection{Analysis of the algorithm}
	We now show that the algorithm gives a $(1+O(\epsilon), 1+O(\epsilon))$-approximation algorithm to the $(k, z)$-median/means problem. %With communication cost $O(\frac{\log \Delta}{\epsilon})$ times the communication cost stated in theorem~\ref{thm:coreset}. 
	With probability at least $1-\delta$, for every $L$, the weighted set $(Q_L, w_L)$ is an $\epsilon$-corset for $P$ w.r.t metric $d_L$.  Let $C^*$ be the optimal set of centers for the original $(k,z)$-median/means problem. Then, for every $z' \in [0, n]$, we have
  \begin{alignat*}{2}
    &\quad \cost'_{\tilde z}({\tilde C}) \\
    &=\  \sup_{L\in \bbL} \left(\sum_{p \in P}d^\ell_L(p, {\tilde C}) - \tilde zL^\ell\right)  &
    &\ \leq\  \sup_{L\in \bbL} \left(\frac{1}{1-\epsilon}\sum_{q \in Q_L}w_q d^\ell_L(q, {\tilde C}) - \tilde zL^\ell\right) \\
    &\leq\ \sup_{L\in \bbL} \left(\frac{1+\epsilon}{1-\epsilon}\sum_{q \in Q_L}w_q d^\ell_L(q, C^*) - \tilde zL^\ell\right)  &
    &\ \leq\ \sup_{L\in \bbL}\left(\frac{(1+\epsilon)^2}{1-\epsilon}\sum_{p \in P}d^\ell_L(p, C^*) - \tilde zL^\ell\right) \\
    &=\ \frac{(1+\epsilon)^2}{1-\epsilon}\sup_{L\in \bbL}\left(\sum_{p \in P}d^\ell_L(p, C^*) - zL^\ell\right) &
    &\ =\ \frac{(1+\epsilon)^2}{1-\epsilon}\cost'_{z}(C^*).
  \end{alignat*}
  The first and the third inequalities are by the definition of $\epsilon$-coreset, while the second inequality is by \eqref{inequ:tilde-C-is-approximate}.
Then with Lemma~\ref{lmm:R-to-bbL}, we know that
  \begin{align*}
   \cost_{\frac{(1+\epsilon)^{\ell+2}}{1-\epsilon}z} (\tilde C) &= \cost_{(1+\epsilon)^\ell\tilde z}({\tilde C}) \leq (1+\epsilon)^\ell\cost'_{\tilde z}({\tilde C})\\
   & \leq \frac{(1+\epsilon)^{\ell + 2}}{1-\epsilon}\cost'_z(C^*) \leq \frac{(1+\epsilon)^{\ell + 2}}{1-\epsilon}\cost_z(C^*).
  \end{align*}
  So, ${\tilde C}$ is a $\left(\frac{(1+\epsilon)^{\ell+2}}{1-\epsilon},\frac{(1+\epsilon)^{\ell+2}}{1-\epsilon}\right) = (1+O(\epsilon), 1+O(\epsilon))$-approximate solution.  We can scale down the input $\epsilon$ by a constant factor to obtain a $(1+\epsilon, 1+\epsilon)$-approximation. 
  
  As we mentioned, the running time of the algorithm for the central coordinator is exponential in $\frac1\epsilon, k, D, m, \log \frac1\delta$ and $\log \Delta$. For each machine $i$, the running time in the algorithm of \cite{DBLP:conf/nips/BalcanEL13} is dominated by the time to compute an $O(1)$-approximation for the $k$-median/$k$-means problem for $P_i$, which is polynomial in $n_i$ and $D$.
	\section{Complete Experiment Results}\label{sec:experiments}
\subsection{$k$-Center Clustering with Outliers}\label{subsec:k-center-exp}
We evaluate the performance of our $(k, z)$-center algorithm (Algorithm~3) on several real-world datasets, which are summarized in Table~\ref{tab:datasets-info}. In the experiments we compare $\distkzcmulti$ with many other $k$-center methods, including two centralized methods (\seqgr\cite{DBLP:journals/mor/HochbaumS85} and $\seqkzc$ \cite{DBLP:conf/soda/CharikarKMN01} and four distributed methods (\distrr, \distrc, \MKCWM \cite{DBLP:conf/nips/MalkomesKCWM15}, and \GLZ \cite{DBLP:conf/spaa/GuhaLZ17}). The \seqgr method has a 2-approximation ratio in the no-outlier scenario, but doesn't take outliers into account.
% The $\seqkzc$ method is also a greedy algorithm which tries to cover as many as possible points with $k$ balls.
The \distrr and \distrc methods serves as two baselines: \distrr randomly sample $k+z$ points on each machine, then further randomly choose $k$ points from the total $m(k+z)$ sampled points as final cluster centers; \distrc is similar to \distrr, except that it chooses the final $k$ centers by the $\seqkzc$ method. The \MKCWM and \GLZ are the state-of-art distributed $k$-center algorithms that handle outliers. For each parameter setting the experiment is repeated for 5 runs and the average result is reported. Note the three distributed baseline methods \distrr, \distrc, and \MKCWM all have the same communication cost $md(k+z)$, while \GLZ's communication cost is $\tldO(mk+m/\epsilon)$. All methods are implemented in Python and the experiments are conducted on a 2-core 2.7 GHz Intel Core i5 laptop.

%We also tried to implement the sampling-based algorithm \distkzcaddi. However, it is very sensitive to the number $n'$ of samples we choose. The performance is unstable if we set $n'$ to be small. When $z$ is relatively small, to get a good result we need $n'$ to be very big. Besides, it is not known ahead of time what is the best $n'$. So we choose not to present the results for \distkzcaddi.

\begin{table}[!htbp]
  \small
  \centering
    \begin{tabular}{rlcc}
      \addlinespace
      \toprule
      \textbf{Name} & \textbf{Size: }$n$ & \textbf{Dimension: }$B$\\
      \midrule
      \textit{spambase}\protect\footnotemark & 4,601  & 57 \\
      \textit{parkinsons}\protect\footnotemark & 5,875 & 16 \\
      \textit{pendigits}\textsuperscript{\ref{fn:uci}} & 10,992  & 16\\
      \textit{letter}\textsuperscript{\ref{fn:uci}}  & 20,000 & 16 \\
      \textit{skin}\textsuperscript{\ref{fn:uci}} & 245,057 & 3 \\
      \textit{covertype}\textsuperscript{\ref{fn:uci}} & 581,012 & 10 \\
      \textit{gas}\textsuperscript{\ref{fn:uci}} & 928,991 & 10 \\
      \textit{power}\textsuperscript{\ref{fn:uci}} & 2,049,280 & 7 \\ 
%      \textit{Higgs}\textsuperscript{\ref{fn:uci}} & 11,000,000& 7\\
      \bottomrule
    \end{tabular}
    \caption{Clustering datasets used for evaluation}
    \label{tab:datasets-info}
\end{table}
\addtocounter{footnote}{-1}
\footnotetext{\label{fn:uci} The UCI data repository \cite{Lichman:2013}}
\addtocounter{footnote}{1}
\footnotetext{\cite{DBLP:conf/icassp/tsanaslmr10}}
The experiments consist of two parts: In the first part we compare our algorithms with the two centralized methods. This part is conducted only for the 4 smaller datasets (\textit{spambase}, \textit{parkinsons}, \textit{pendigits}, and \textit{letter}), on which centralized methods can finish in an acceptable time. In the second part we compare our algorithms with other distributed methods on the 4 larger datasets (\textit{skin}, \textit{covertype}, \textit{gas}, and \textit{power}). 

\noindent\textbf{Distributed v.s. Centralized: } Figure~\ref{fig:seq-vs-dist-var-z} and Figure~\ref{fig:seq-vs-dist-var-k} show the results on the four smaller datasets. Figure~\ref{fig:seq-vs-dist-var-z} demonstrates how the objective value and communication cost change with $z$ when $k$ is fixed to $20$. Our algorithm $\distkzcmulti$ always achieve comparable objective with other distributed baselines. On datasets \textit{spambase} and \textit{parkinsons}, the objective value even matches the best centralized method ($\seqkzc$). When it comes to communication cost, $\distkzcmulti$ shows a clear advantage over \distrr, \distrc, and \MKCWM, which matches our theoretical results.

Figure~\ref{fig:seq-vs-dist-var-k} depicts the performance with respect to different value of $k$ when $z$ is fixed to $256$. $\distkzcmulti$ still achieves similar (or better) objective values among all distributed methods. But we can see that when $k$ increases, the communication cost of $\distkzcmulti$ ($\epsilon=0.1$) approaches those of other distributed methods. Recall that the communication cost of $\distkzcmulti$ is $\tldO(mk/\epsilon)$ which can be similar to $O(m(k+z))$ when $z$ and $k/\epsilon$ are in the same order. If we choose a large value of $\epsilon=0.99$, the communication cost of $\distkzcmulti$ becomes much stable, while the objective value is only slightly worse. This suggests that in practice we can choose a relatively large $\epsilon$ to obtain small communication cost.

We want to remind the readers that the approximation ratio of $\distkzcmulti$ holds for removing $(1+\epsilon)z$ outliers, while in the experiments the objective is computed by removing only $z$ outliers. This indicates that $\distkzcmulti$ may have better performance than what is predicted theoretically.

\begin{figure}[!htb]
  \centering
    \includegraphics[width=\textwidth]{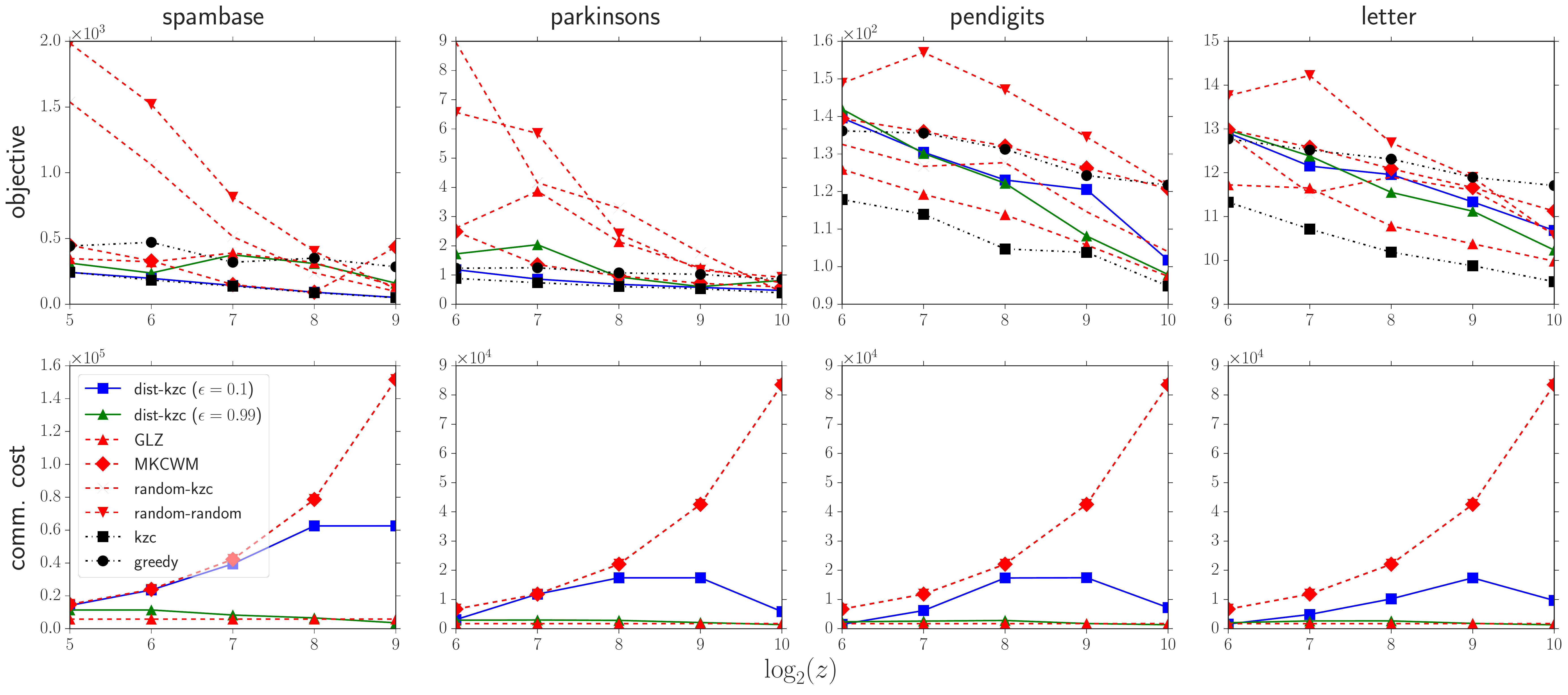}
  \caption{Centralized vs. Distributed, with varying $z$ and fixed $k=20, m=5$.}
  \label{fig:seq-vs-dist-var-z}
\end{figure}

\begin{figure}[!htb]
  \centering
    \includegraphics[width=\textwidth]{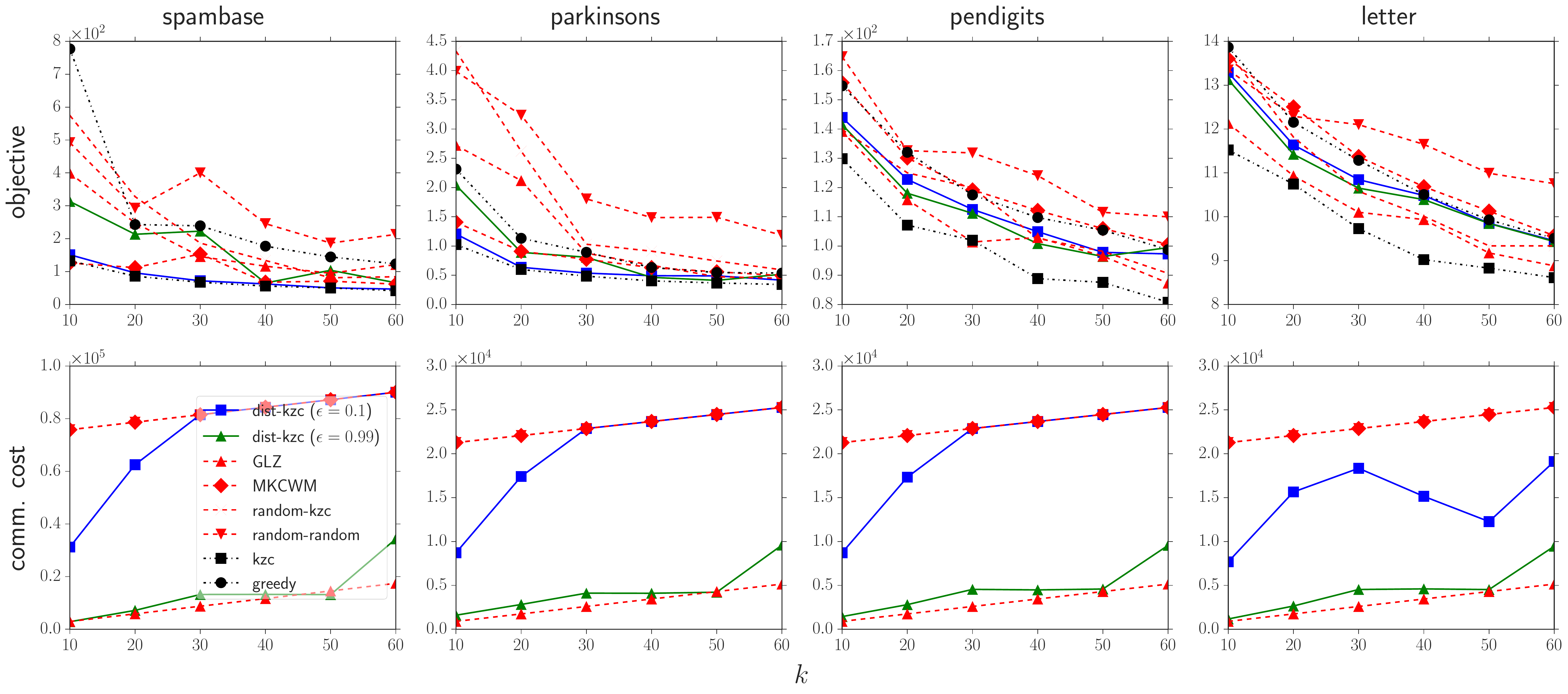}
  \caption{Centralized vs. Distributed,  with varying $k$ and fixed $z=256, m=5$.}
  \label{fig:seq-vs-dist-var-k}
\end{figure}

\noindent\textbf{Large scale: }This part contains experiment results on the four large datasets: \textit{skin}, \textit{covertype}, \textit{gas}, and \textit{power}. The \GLZ\ method needs solving many local $(k,z')$-center instances, which is too slow to finish on these large datasets. Hence here we use its variant provided by \cite{DBLP:conf/spaa/GuhaLZ17}, denoted as \GLZz. \GLZz works similar as \GLZ, but avoids solving $(k,z')$-center locally on each machine by transmitting $\tldO(mk+z)$ data to the coordinator. So \GLZz has a higher communication cost than \GLZ, but it's still much better than \MKCWM which has a $O(m(k+z))$ communication cost.

Similar to the previous part, Figure~\ref{fig:large-scale-var-z} and Figure~\ref{fig:large-scale-var-k} show results for varying $z$ and $k$. Our method still achieves comparable objective value with the best distributed baselines. The communication cost of our algorithm is always much smaller than \MKCWM, and matches that of \GLZz. This advantage is more obvious with bigger $z$, but here to make all the baselines terminate in acceptable times we only use $z\sim\sqrt{n}$.

\begin{figure}[!htb]
  \centering
    \includegraphics[width=\textwidth]{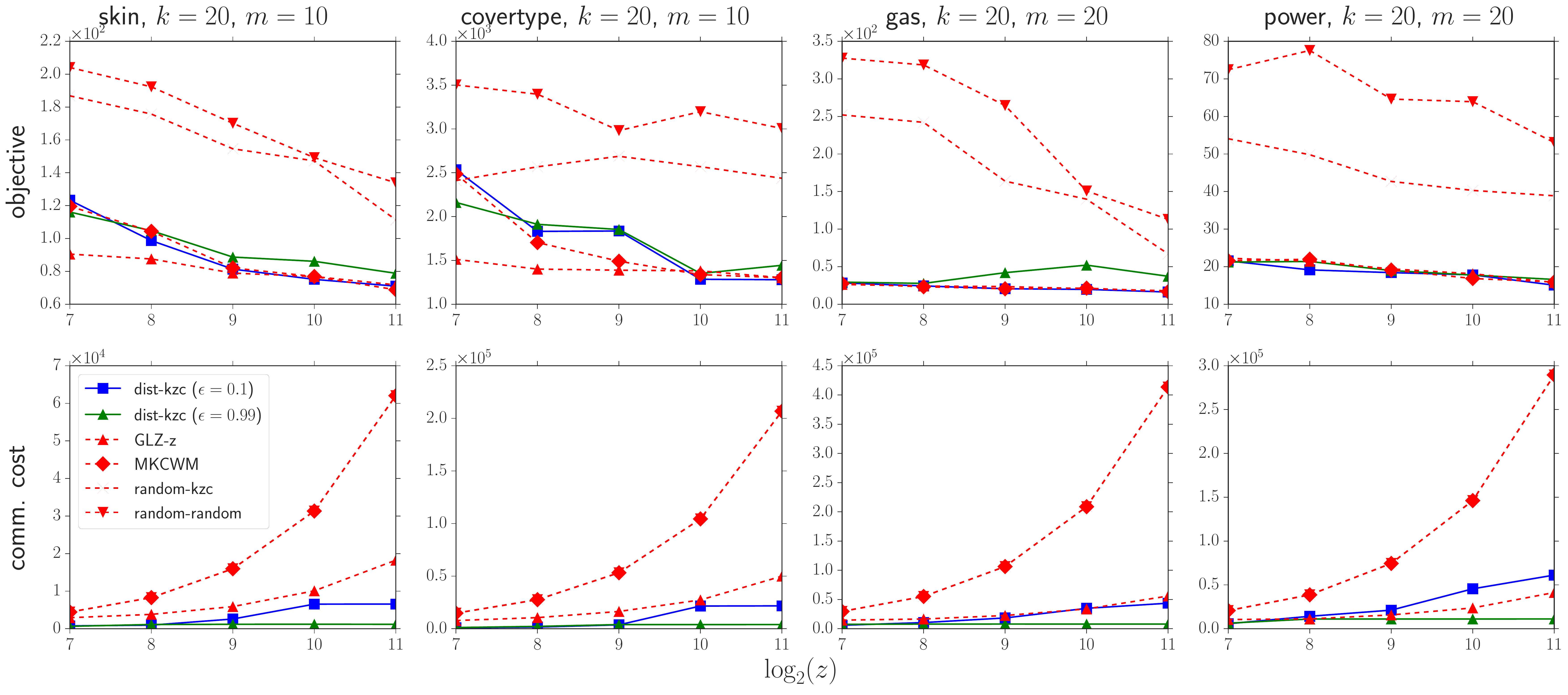}
  \caption{Large scale, with varying $z$}
  \label{fig:large-scale-var-z}
\end{figure}

\begin{figure}[!htb]
  \centering
    \includegraphics[width=\textwidth]{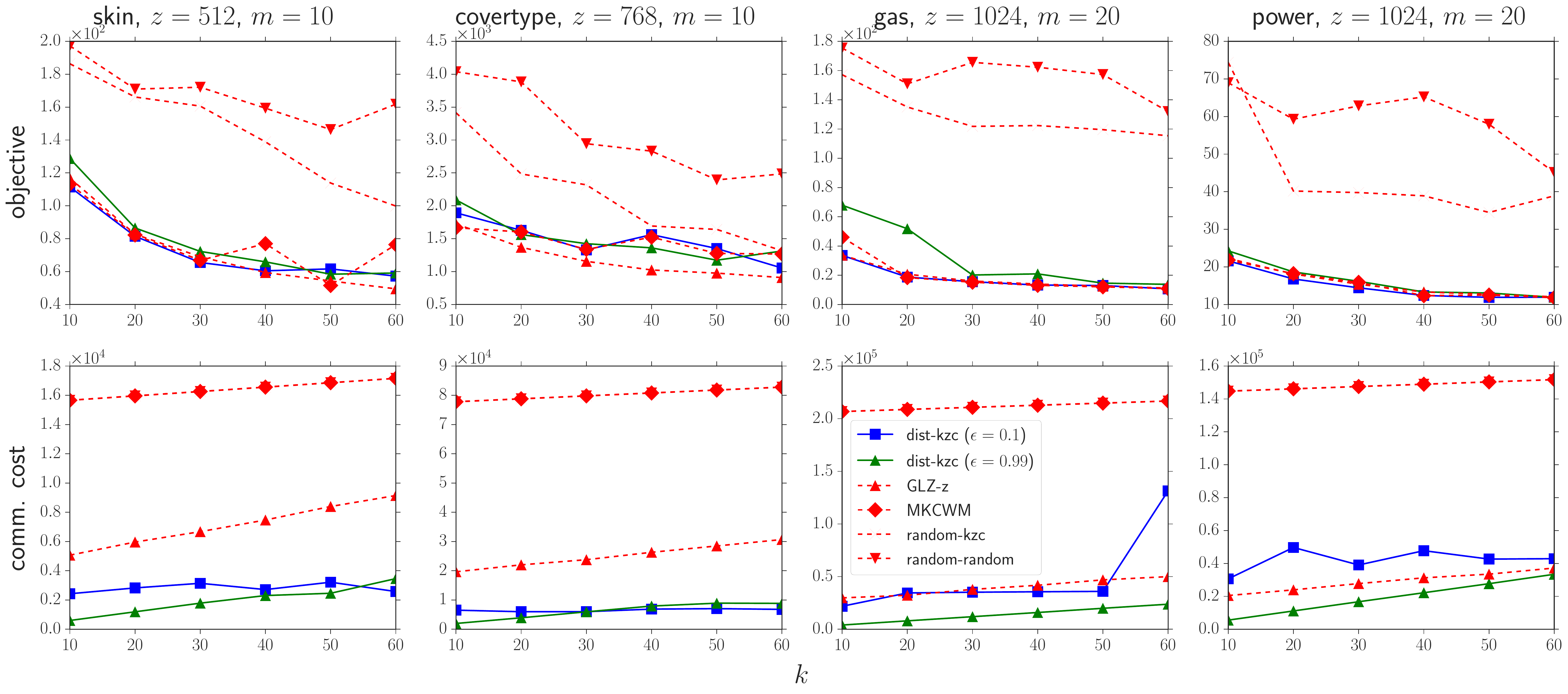}
  \caption{Large scale, with varying $k$}
  \label{fig:large-scale-var-k}
\end{figure}

% For the two largest data sets (\textit{Gas} and \textit{Power}), we also experiment with $z$ as large as $2^{14}=16,384$ (which is still only $n/100$, though). The three distributed baselines \GLZ, \MKCWM, and \distrc run too slowly for this setting so we only compare our algorithm with the \distrr method. Since the communication cost of the three distributed baselines are always the same, we can still see how \MKCWM and \distrc compare with our algorithm in terms of communication cost. The results are shown in Figure~\ref{fig:very-large-z}.
% \makeatletter
% %\setlength{\@fptop}{12pt}
% %\setlength{\@fpbot}{0pt plus 1fil}
% %\makeatother
% \begin{figure}[!htbp]
%   \centering
%   \subfloat{
%     \includegraphics[width=0.52\textwidth]{figs/large_z-gas-nips.pdf}
%   }
%   ~
%   \subfloat{
%     \includegraphics[width=0.52\textwidth]{figs/large_z-power-nips.pdf}
%   }
%   % ~
%   % \subfloat[\textit{Higgs}]{
%   %   \includegraphics[width=0.26\textwidth]{figs/letter-var-z-k-20-v.pdf}
%   % }
%   \caption{Experiment of large $z=2^{14}=16384$ and varying $k$, with $m=20$}
%   \label{fig:very-large-z}
% \end{figure}

        \subsection{$k$-Means Clustering with Outliers}
\noindent\textbf{The centralized solver:} We test our distributed $k$-means algorithm proposed in section~\ref{sec:app-kzmedian}. As we described, the algorithm requires solving a \emph{min-max} clustering problem on the coordinator. Formally, given a set of datasets $Q_1,\ldots,Q_M$, each equipped with its own metric $d_1,\ldots,d_M$, the goal is to find a center set $C$ minimizing the maximum cost over all $M$ datasets:
\begin{equation}
  \min_{C:|C|=k}\max_{i\in[M]}\sum_{p\in Q_i} d^l_i(p, Q_i)
\end{equation}
where $l=1, 2$ corresponding to the $k$-median or $k$-means objective respectively. 

Although we don't know any practical algorithm for such min-max clustering problem, there exists some results addressing a simpler form of the min-max $k$-median problem: Suppose there're only $N$ possible locations for selecting the center set $F$ (i.e., $C\subset V$ for some $|V|=N$), and every dataset $Q_i$ has the same metric $d$, then Anthony~\etal\cite{DBLP:journals/mor/AnthonyGGN10} shows that a simple reverse-greedy method achieves $O(\log N+\log M)$-approximation for the min-max $k$-median problem in this special case. We adapt their method to solve our min-max $k$-means problem in the experiment. For completeness, the algorithm is listed below:
\begin{algorithm}[H]
  \caption{\rgreedy$(k, \{(Q_i,d_i)\}_{i=1}^M, B)$\cite{DBLP:journals/mor/AnthonyGGN10}} \label{alg:anthony}
  \begin{algorithmic}[1]
    \STATE $C^1\gets \bigcup_{i=1}^MQ_i, w^1_i\longleftarrow 1$ for all $i\in[M]$;
    \STATE \textbf{for} {$t \gets 1$ to $N-k$} \textbf{do}
    \STATE \hspace*{\algorithmicindent} For every $v\in F^t$ and $i\in[M]$, let $\delta^t_i(v)\gets\sum_{p\in Q_i}(d^2_i(p, C^t\setminus\{v\})-d^2_i(p, C^t))$
    \STATE \hspace*{\algorithmicindent} $\hat{C}^t \gets \{v\in C^t|\forall i\in[M], \delta^t_i(v)\leq B/2\}$\;
    \STATE \hspace*{\algorithmicindent} $v^t\gets\argmin_{v\in \hat{C}^t}\sum_{i=1}^Mw^t_i\cdot\delta^t_i(v)$ \; 
    \STATE \hspace*{\algorithmicindent} For all $i\in[M]$, let $w^{t+1}_i\gets w^t_i\left(1+\frac{1}{B}\right)^{\delta^t_i(v^t)}$\;
    \STATE \textbf{return} $C^{N-k+1}$
  \end{algorithmic}
\end{algorithm}
Roughly speaking, the algorithm starts with $C$ being the set of all points, and iteratively remove points in $C$ until it shrinks to size $k$. In each iteration the algorithm removes from $C$ the point that incurs the least weighted total cost increase. However, because our problem is more general than that in \cite{DBLP:journals/mor/AnthonyGGN10}, we don't know whether their approximation guarantee for Algorithm~\ref{alg:anthony} still holds here. 

\noindent\textbf{Algorithms:} We compare our implementation with some other algorithms for the $k$-means/$(k,z)$-means problem, including two centralized ones and two distributed ones: \kmeans\cite{Lloyd06}, the classical Lloyd's algorithm; \kmeansmm\cite{DBLP:conf/sdm/ChawlaG13}, like \kmeans, but uses some heuristics to handle outliers; \BEL\cite{DBLP:conf/nips/BalcanEL13}, the distributed $k$-means algorithm based on coreset; and \CAZ\cite{DBLP:journals/corr/abs-1805-09495}, a recently proposed distributed $(k,z)$-means algorithm. The \BEL and \CAZ algorithms both belong to the two-level clustering framework\cite{DBLP:journals/tkde/GuhaMMMO03}: first construct a local summary on each machine and aggregate them on the coordinator, then the coordinator conduct a centralized clustering over the aggregated summaries to get the final result. But the main focus of \BEL and \CAZ is how to construct local summary, and they don't specify the actual coordinator solver used. In the experiment we use \kmeans and \kmeansmm as the centralized solver for \BEL and \CAZ respectively. All methods are implemented in Python and the experiments are conducted on a 2-core 2.7 GHz Intel Core i5 laptop.

\noindent\textbf{Datasets:} The experiment is conducted on one synthesized dataset and three real-world datasets. The real-world datasets are \textit{spambase, parkinsons}, and \textit{pendigits} (see Table~\ref{tab:datasets-info}). Unlike the $k$-center case, the outliers in the original dataset are unable to significantly affect the objective value. Thus to make the algorithm's effect clearer, we manually add 500 outlier points to each of the three dataset. The synthesized dataset is sampled from a mixture of Gaussian model, of which the parameters are also randomly generated; specifically, we sample 10000 points  in total from 4 different Gaussian distributions in $\R^5$, and manually add another 500 outliers to the dataset.

\noindent\textbf{Parameter setting:} For each dataset, we fix $k$ and vary $z$. On the three real-world datasets, $k$ is set to be $10$ and $z$ varies from $2^5$ to $2^{11}$; on the synthesized dataset, $k$ is set to be $4$ and $z$ ranges from $2^6$ to $2^{10}$. The number of machines are fixed to $5$ for all 4 datasets. Throughout the experiment, we use $\epsilon=0.3$ as the error parameter for our algorithm. We measure how the objective value and communication cost (for distributed methods only) changes with $z$. But different from the setting in Section~\ref{subsec:k-center-exp}, here we compute the cost of our method by removing $(1+\epsilon)z$ outliers to match our theory result. (In this sense, the comparison is ``more fair'' for us than in Section~\ref{subsec:k-center-exp}) 

Another issue in applying our $(k,z)$-means algorithm is the choice of appropriate coreset size. Unlike the result for our $(k,z)$-center algorithm, we only have an asymptotic estimation for the coreset size, which is not so instructive in practice. Therefore, in the experiment we hand-pick the coreset size by some heuristics: when the value of the error parameter $\epsilon$ is given, we can compute the total number of different threshold distance that will be tried (i.e., $|\bbL|$). Then we choose the coreset size to be $\max\left\{10k,\frac{n}{10m|\bbL|}\right\}$. So each coreset contains at least $10k$ samples, and when $n\gg km|\bbL|$, we allow the total size to be as large as $n/10$.
\begin{figure}[!htb]
  \centering
    \includegraphics[width=\textwidth]{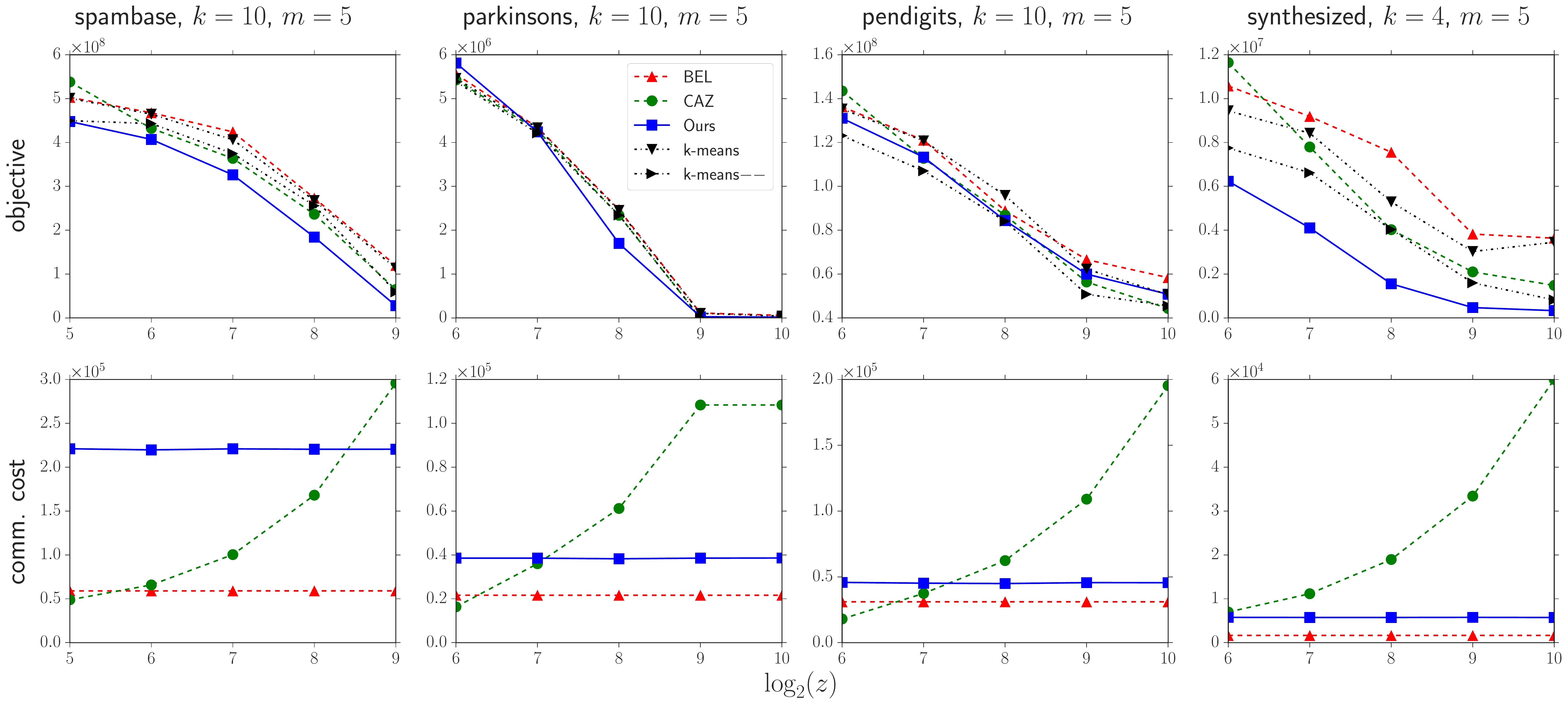}
  \caption{Comparison of the our distributed $(k,z)$-means implementation with other distributed/centralized methods. The first row is objective value, and the second is communication cost.}
  \label{fig:kzmeans-results}
\end{figure}

\noindent\textbf{Experiment results:} Figure~\ref{fig:kzmeans-results} shows the experiment result: we can see that our algorithm performs surprisingly well in terms of objective value, often achieving the lowest cost among all the methods. The effect of outliers is most clearly revealed on the synthesized data, where \BEL and \kmeans perform significantly worse than others. In particular, although we remove $\epsilon z$ more outliers when calculating the cost for our method, it's still much better than \BEL even if compared at different $z$: consider our method's cost at $z=(1+\epsilon)2^7=1.3\cdot2^7$ with \BEL's at $z=2^8$. 

The communication cost of our method doesn't change with $z$, since the way we decide the coreset size makes it fixed. \BEL's communication cost is also not affected by $z$, as it doesn't deal with outliers. In contrast, \CAZ's communication is in the order of $O(mk\log n + z)$, which is reflected in the figure as it grows linearly in $z$. Although our centralized solver uses some heuristics and thus doesn't have provable guarantees, the experiment results suggest that our coresets construction indeed preserves the outliers information while being independent of $z$.

\end{document}